%% file: Counting_Holant_Polymers.tex
\documentclass
[
a4paper,
11pt
]
{scrartcl}

\usepackage[utf8x]{inputenc}
\usepackage[T1]{fontenc}
\input{preamble.tex}

\title{Zeros and approximations of Holant polynomials on the complex plane}
\author{\small Katrin~Casel\thanks{\texttt{
			\{katrin.casel, philipp.fischbeck, tobias.friedrich, andreas.goebel, gregor.lagodzinski\}@hpi.de}
		}
	}	
\author{\small Philipp~Fischbeck\thanksmark{1} }
\author{\small Tobias~Friedrich\thanksmark{1} }
\author{\small Andreas~G\"obel\thanksmark{1} }
\author{\small J.~A.~Gregor~Lagodzinski\thanksmark{1} }

\affil{\small Hasso Plattner Institute, University of Potsdam, Potsdam, Germany }

\date{}

\begin{document}

\maketitle

\begin{abstract}
	We present fully polynomial-time approximation schemes for a broad class of Holant problems with complex edge weights, which we call \emph{Holant polynomials}. 
	We transform these problems into partition functions of abstract combinatorial structures known as \emph{polymers} in statistical physics. Our method involves establishing zero-free regions for the partition functions of polymer models and using the most significant terms of the \emph{cluster expansion} to approximate them.
	
	Results of our technique include new approximation and sampling algorithms for a diverse class of Holant polynomials in the low-temperature regime and approximation algorithms for general Holant problems with small signature weights.	
	Additionally, we give randomised approximation and sampling algorithms with faster running times for more restrictive classes. Finally, we improve the known zero-free regions for a perfect matching polynomial.
	
\end{abstract}	

\section{Introduction}
\label{sec:intro}
\input{intro_soda.tex}

\section{Boolean Holant polynomials}
\label{sec:intro-Boolean}
\input{boolean.tex}

\section{Abstract polymer models and approximation algorithms}
\label{sec:polymers}
\input{polymers.tex}

\section{Holants as polymer models}
\label{sec:holant_polynomials}
\input{holant_as_polymers.tex}

\section{Deterministic algorithms}
\label{sec:fptas}
\input{holant_deterministic.tex}

\section{Fast randomised algorithms}
\label{sec:random}
\input{random.tex}

\section{Further applications of abstract polymers}
\label{sec:outlook}
\input{outlook.tex}

\section{Acknowledgements}
The authors would like to thank Heng Guo and Guus Regts for helpful comments on an early draft of this work.

\bibliographystyle{plain}
\bibliography{references}
	
\end{document}

%% file: preamble.tex
\usepackage{authblk}
\DeclareRobustCommand{\thanksmark}[1]{\footnotemark[#1]}

\usepackage{lmodern}
\usepackage{datetime}
\usepackage{hyperref}
\usepackage{listings}
\usepackage[greek,english]{babel}
\usepackage[sort,numbers]{natbib}

\usepackage{xparse} 

\usepackage{graphicx}
\usepackage{enumerate}
\usepackage{enumitem}
\usepackage{longtable}
\usepackage{amsfonts,amssymb,amstext,amsthm,amsmath}
\usepackage{mathabx}
\usepackage[mathscr]{euscript}
\usepackage{bbm}  
\usepackage{tikz}
\usetikzlibrary{decorations.pathmorphing}
\usetikzlibrary{decorations.pathreplacing} 
\usetikzlibrary{calc}
\usetikzlibrary{shapes.symbols}
\usetikzlibrary{arrows} 
\usetikzlibrary{decorations.markings}

\usepackage{color,soul}

\usepackage{calc}
\usepackage{filecontents}
\usepackage{pifont}
\usepackage{dsfont}

\allowdisplaybreaks

\usepackage{xcolor}

\definecolor{mycolor1}{RGB}{0,5,200}
\definecolor{mycolor2}{RGB}{0,150,5}

\newcommand{\ex}{\mathrm{e}}

\newcommand{\maxValue}{\ensuremath{r}}
\newcommand{\maxValueOne}{\ensuremath{r_1}}

\newcommand{\classfont}{\mathsf}
\newcommand{\class}[1]{\mbox{{\(\classfont{\,#1}\)}}}

\newcommand{\np}{\ensuremath{\class{NP}}}
\newcommand{\nump}{\ensuremath{\#\class{P}}}
\newcommand{\adjedges}{\ensuremath{E}}

\newcommand{\Holf}{\ensuremath{\mathrm{Hol}_\SetFunctions}}

\newcommand{\N}{\ensuremath{\mathbb{N}}}
\newcommand{\Z}{\ensuremath{\mathbb{Z}}}
\newcommand{\C}{\ensuremath{\mathbb{C}}}
\newcommand{\R}{\ensuremath{\mathbb{R}}}
\newcommand{\E}{\ensuremath{\mathbb{E}}}
\renewcommand{\P}{\ensuremath{\mathbb{P}}}

\newcommand{\SetFunctions}{\ensuremath{\mathcal{F}}}
\newcommand{\Matchings}{\ensuremath{\mathcal{M}}}

\newcommand{\Independent}{\ensuremath{\mathcal{I}}}

\newcommand{\dist}{\ensuremath{\mathrm{dist}}}

\DeclareMathOperator{\supp}{supp}

\newcommand{\vdompoly}{\ensuremath{\gamma_v}}
\newcommand{\vcolouring}{\ensuremath{\varphi_\gamma}}
\newcommand{\vSetPolymers}{\mathcal C _v}
\newcommand{\vfunction}{\ensuremath{Z(\vSetPolymers(G),\Phi(\cdot,z))}}

\newcommand{\groundmatching}{\ensuremath{M}}
\newcommand{\difference}{\ensuremath{\Delta}}

\newcommand{\domainvarfunctions}{\ensuremath{\mathcal F_0}}
\newcommand{\domain}{\ensuremath{D}}
\newcommand{\domainvar}{\ensuremath{\kappa}}
\newcommand{\adjkappa}{\ensuremath{E^\domainvar}}

\newcommand{\dompoly}{\ensuremath{\gamma_\phi}}

\theoremstyle{plain}
\newtheorem{theorem}{Theorem}[section]
\newtheorem{corollary}[theorem]{Corollary}
\newtheorem{lemma}[theorem]{Lemma}

\theoremstyle{definition}
\newtheorem{definition}[theorem]{Definition}

\newtheorem{remark}[theorem]{Remark}
\newtheorem{observation}[theorem]{Observation}

%% file: intro_soda.tex

The ongoing effort to characterise the complexity of approximating partition functions originating from statistical physics has recently seen great progress. Such partition functions take as input a metric, usually a graph, that encodes how particles interact. Two kinds of such partition functions have been studied in the literature: vertex spin systems, where vertices assume some state (spin) and edges are functions encoding the energy of vertex interactions; and edge spin systems, where the edges assume spins and the vertices are functions encoding the energy of edge interactions. 
So far, the main body of literature focuses on the first category of vertex spin systems. Results include remarkable connections between the phase transitions spin systems undergo and the approximability of their partition function, see e.g.~\cite{2006:Weitz:counting_independent_sets,2010:Sly:computational_transition,2017:Peters:conjecture_of_Sokal,2014:Galanis:inapproximability_independent_hard_core}. 

This article focuses on the second category of edge spin systems, which can be naturally encoded under the \emph{Holant} framework of counting problems. Our results include approximation algorithms for \emph{Holant polynomials} in the low energy regime and approximation algorithms for \emph{Holant problems} with vertices encoding small weights. In particular we identify new tractable cases for such Holants with respect to their approximability.

The Holant framework originates in Valiant's holographic algorithms~\cite{2008:Valiant:holographic_algorithms} to model perfect matching computations and was extended by Cai, Lu and Xia~\cite{2011:Cai:complexity_Holant} to encode partition functions of edge spin systems such as edge covers and Eulerian orientations. A Holant problem is parametrised by a finite set of functions (often called \emph{signatures}) $\SetFunctions$ with domain $D$.
For a graph $G=(V,E)$ and a mapping $\pi\colon V\rightarrow\SetFunctions$, where $\pi(v)=f_v$ maps $v\in V$ to a function $f_v\colon D^{\adjedges(v)}\rightarrow \mathbb C$, the \emph{Holant problem} is to compute the function
\[
 \Holf(G,\pi)=\sum_{\sigma\in D^E} \prod_{v\in V}f_v(\sigma|_{\adjedges(v)})
\]
with $\adjedges(v)$ denoting the set of edges incident to $v$.

External conditions like fugacity are known in statistical physics as \emph{external fields} and encoded in partition functions as weights in the system. In the Holant framework this yields Holant problems with external fields, which we call \emph{Holant polynomials}. For a signature set $\SetFunctions$ and a domain $D=\{0,1,\dots,\kappa\}$ for some $\kappa\in\Z_{>0}$ the \emph{Holant polynomial} maps a graph $G=(V,E)$ and a mapping $\pi\colon V\rightarrow\SetFunctions$
to the function
\[
	Z_\SetFunctions(G,\pi,\mathbf{z})=\sum_{\sigma\in D^{E}} \prod_{v\in V}f_v(\sigma|_{\adjedges(v)})\prod_{i=0}^\kappa z_i^{|\sigma|_i},
\]
where $\mathbf{z}=(z_0,z_1,\dots,z_\kappa)$ and $|\sigma|_i$ denotes the number of edges $e\in E$ with $\sigma(e)=i$. As an example, if $\mathbf{z}=(1,z)$ and the set of signatures $\SetFunctions$ contains the functions taking the value~1 if at most one of the input variables is set to~1, the Holant polynomial is the matching polynomial $Z_{\Matchings}(G,\mathbf{z})$.

Our results restrict the signature sets to be subsets of $\SetFunctions_0=\{f\mid f(\mathbf{0})\neq 0\}$ and to contain only polynomially computable signatures. For $f\in\domainvarfunctions$ with arity $d$ we define $r(f)=\max_{x\in\domain^d\setminus\{\mathbf{0}\}}\{|f(x)|/|f(\mathbf{0})|\}$ and for a signature class $\SetFunctions\subseteq\domainvarfunctions$ we define $r(\SetFunctions)=\max_{f\in\SetFunctions}\{r(f)\}$. Our main theorem is the following.
\begin{theorem}
	\label{thm:Holant_poly_FPTAS}
	Let $\mathcal F\subseteq \mathcal F_0$. For all graphs $G$ of maximum degree $\Delta$ and all $\pi\colon V\rightarrow \mathcal F$, the Holant polynomial admits an FPTAS for $\mathbf{z}$ in
	\[\left\{(z_0,z_1,\dots,z_\domainvar)\in \mathbb C^{\domainvar+1}\ \Big|\  z_0\neq0, \frac{|z_i|}{|z_0|}<  (\Delta\domainvar\ex^{2}\maxValueOne(\maxValueOne+1))^{-1},\ 1\leq i\leq \domainvar\right\}\]
	with $\maxValueOne = \max\{1, r(\mathcal F)\}$.
\end{theorem}

This result captures a very broad class of Holant polynomials as the only requirement in terms of the class of signatures is that the all zeros assignment $\sigma_0$ contributes a non-zero term in the sum. Essentially, this ensures the existence of a trivially computable solution for the related decision problem, the all zeros configuration. In a sense, an efficient way to show the existence of a solution is necessary since hardness of this decision problem immediately implies hardness of approximation as Goldberg et al.~observe~\cite[Theorem~1]{2004:Goldberg:relative_complexity_approx_coutning}.

The region stated in Theorem~\ref{thm:Holant_poly_FPTAS} excludes the value $\mathbf{z}=\mathbf{1}$, therefore it does not apply to Holant problems directly. However, our technique also yields the following theorem.
\begin{theorem}\label{thm:Holant_FPTAS}
 Let $\mathcal F\subseteq \mathcal F_0$, with $\maxValue(\SetFunctions)< \max\{(2\sqrt\ex)^{-1}(\Delta\kappa\ex)^{-\frac\Delta2} ,0.2058(\domainvar+1)^{-\Delta}\}$. For all graphs $G$ of maximum degree $\Delta$ and all $\pi\colon V\rightarrow \mathcal F$, the Holant problem $Z_{\mathcal F}(G,\pi)$ admits an FPTAS.
\end{theorem}

Key to our technique is to translate the Holant polynomial to the partition function of an \emph{abstract polymer model}. One of the advantages of polymer models is that they are \emph{self-reducible}. A known implication of self-reducibility is the equivalence between approximate counting and sampling~\cite{Jerrum:1986:Random_Generation, Sinclair:1989:Approximate_Counting}. As Helmuth, Perkins and Regts observe~\cite[Section~5]{2018:Helmuth:algorithmic_pirogov_sinai_theory} counting algorithms for polymer models can be converted to algorithms that approximately sample from the \emph{Gibbs distribution}~$\mu_G$ where in our case
\[
	\mu_G(\sigma)=\frac{\prod_{v\in V(G)}f_v(\sigma|_{\adjedges(v)})\prod_{i=0}^\kappa z_i^{|\sigma|_i}}{Z_{\mathcal F}(G,\pi,\mathbf{z})}.
\]

One of the downsides of the above algorithms is that the runtime is in $\mathcal O(n^\Delta)$. Under more restrictive conditions and for particular polymer models Chen et al.~\cite{2019:Chen:fast_algorithms_low_temperature} showed how to obtain faster randomised algorithms. We extend these results to show that they apply to our polymer models for Holant polynomials.
\begin{theorem}\label{thm:holant_poly_random}
 Let $G$ be a graph of maximum degree $\Delta$, $\mathcal F\subseteq\mathcal F_0$ and $\pi\colon V(G)\rightarrow \mathcal F$. For the Holant polynomial $Z_{\mathcal F}(G,\pi,\mathbf{z})$ there exists an
 $\varepsilon$-sampling algorithm from the distribution $\mu_G$ with run-time in $\mathcal O(\Delta n\log(n/\varepsilon))$ for $\mathbf{z}$ in
 \[\left\{(z_0,z_1,\dots,z_\domainvar)\in (\R_{\geq 0})^{\domainvar+1}\ \Big|\  z_0\neq0, \frac{z_i}{z_0}\leq ((\Delta\domainvar)^3\ex^5\maxValueOne^2)^{-1},\ 1\leq i\leq \domainvar\right\}\]
  with $\maxValueOne=\max\{1, r(\mathcal F)\}$. Furthermore, $Z_{\mathcal F}(G,\pi,\mathbf{z})$ admits an FPRAS with run-time $\mathcal O(\Delta n^2/\varepsilon^2\log^2(\Delta n/\varepsilon))$ for these values of $\mathbf{z}$.
\end{theorem}

Again, our technique also applies to Holant problems.
\begin{theorem}\label{thm:holant_random}
 Let $G$ be a graph of maximum degree $\Delta$, $\mathcal F\subseteq\mathcal F_0$ such that $r(\mathcal{F})\leq ((\Delta\domainvar)^{-\frac{3\Delta}{2}}\ex^{-\frac{5\Delta}{2}}) $ and $\pi\colon V(G)\rightarrow \mathcal F$. There exists an $\varepsilon$-sampling algorithm from the distribution $\mu_G$ for $Z_{\mathcal F}(G,\pi)$ with run-time in $\mathcal O(\Delta n\log(n/\varepsilon))$. Furthermore, $Z_{\mathcal F}(G,\pi)$ admits an FPRAS with run-time $\mathcal O(\Delta n^2/\varepsilon^2\log^2(\Delta n/\varepsilon))$.
\end{theorem}

\subsection{Methodology}
\label{sec:intro_method}
The central part of our technical approach is to establish a polymer model for Holant problems and Holant polynomials, which translates these into independent set problems. Polymer models are an established tool for the study of partition functions in statistical physics originating from the work of Gruber and Kunz~\cite{1971:Gruber:properties_polymer_systems} and Koteck{\'y} and Preiss~\cite{1986:Kotecky:cluster_expansion_polymer_models}. A polymer model consists of a finite set~$K$ of elements called \emph{polymers} and a symmetric and reflexive binary \emph{incompatibility} relation denoted by~$\not\sim~\subseteq K\times K$. Based on this relation, $\mathcal I(K)$ denotes the set of all subsets, $\Gamma \subseteq K$ of pairwise compatible polymers, which we will call \emph{families} of pairwise compatible polymers.
Given an assignment of weights (\emph{polymer functionals}) $\Phi\colon K\rightarrow \mathbb C$ to the polymers yields the \emph{polymer partition function} $Z(K,\Phi)$ by 
\[
	Z(K,\Phi)=\sum_{\Gamma\in \mathcal I(K)}\prod_{\gamma \in \Gamma} \Phi(\gamma)\,.
\]

Given a polymer model we can construct the \emph{polymer graph} $(K,\not\sim)$ where the polymers represent the vertices of this graph and the edges are given by the incompatibility relation. In this way we observe that the families of compatible polymers are the independent sets of the polymer graph. Weighted sums of independent sets are naturally expressed by the independence polynomial of a graph~$G$ defined on $\mathbf{z}=(z_v)_{v\in V(G)}$ as
\[
	Z_\Independent(G,\mathbf{z})=\sum_{I\in\Independent(G)}\prod_{v\in I}z_v ,
\] 
where $\Independent(G)$ is the set of independent sets of~$G$. From this definition observe that the partition function of a polymer system is the independence polynomial of the polymer graph, where each vertex~$\gamma$ in the polymer graph has weight $\Phi(\gamma)$.

As a simple example of how to translate a Holant polynomial to a polymer system consider again the matching polynomial~$Z_\Matchings$. 
By converting $G$ to its line graph $G'$ observe that 
$Z_\Matchings(G,(1,z))=Z_\Independent(G',\mathbf{z})$, where $\mathbf{z}=(z)_{v\in V(G)}$. Our general method of Holant polynomials applied to the matching polynomial precisely captures this conversion:
the polymer graph $(K,\not\sim)$ is the line graph~$G'$ of $G$, polymers are the edges of $G$ and the weight function is $\Phi(\gamma)=z$. 

In Section~\ref{sec:holant_polynomials} we show how this translation can be extended to a general class of Holant polynomials.

\paragraph{Deterministic Algorithms.}

The strategy to derive a multiplicative approximation for $Z_\SetFunctions$ is to obtain an additive approximation for $\log Z(G,\Phi)$. This idea originates from the work of Mayer and Montroll~\cite{1941:Mayer:molecular_distribution}, which gave a convenient infinite series representation of $\log Z$ called the \emph{cluster expansion} and observed that one can obtain good evaluations of $\lim_{V(G)\rightarrow\infty} \frac{1}{V(G)}\log Z(G,\Phi)$ using the cluster expansion of $\log Z(G,\Phi)$.

There are two main advantages of translating a partition function to a polymer system. The first one is the convenient-to-use theorem of Koteck{\'y} and Preiss \cite[Theorem~1]{1986:Kotecky:cluster_expansion_polymer_models} that establishes conditions for zero-free regions of $Z(G,\Phi)$ and absolute convergence of the cluster expansion of $\log Z(G,\Phi)$.
The second advantage is that when the cluster expansion of $\log Z(G,\Phi)$ converges absolutely then it is the same series as the Taylor series expansion of $\log Z(G,\Phi)$, as Dobrushin observes~\cite{1996:Dobrushin:Estimates_of_semiinvariants}. This allows us to efficiently compute the first terms of the Taylor expansion of $\log Z(G,\Phi)$ using the cluster expansion and obtain a sufficient additive approximation for $\log Z(G,\Phi)$. 

The first ones to derive approximation schemes employing polymer systems were Helmuth et al.~\cite{2018:Helmuth:algorithmic_pirogov_sinai_theory}. With~\cite[Theorem~2.2]{2018:Helmuth:algorithmic_pirogov_sinai_theory} they give conditions under which such a system can be used to obtain deterministic approximation algorithms. Their theorem only applies to polymer systems for vertex spin systems and not to our polymer model. To this end, we extend their theorem and give conditions under which partition functions of general polymer systems can be efficiently approximated (Theorem~\ref{thm:polymers_fptas}). We remark that Theorem~\ref{thm:polymers_fptas} is not restricted to polymers for Holants and it might be of independent interest.

\paragraph{Randomised Algorithms.}

Our fast randomised algorithms are based on the \emph{Markov chain Monte Carlo method} (MCMC). Given a polymer system $(K,\not\sim)$ with weight function $\Phi$ we define a Markov chain with $\mathcal I(K)$ as state space and stationary distribution $\mu_K$, where
\[\mu_K(\Gamma)=\frac{\prod_{\gamma\in\Gamma}\Phi(\gamma)}{Z(K,\Phi)}.\]
We identify the conditions that constitute this chain as \emph{rapidly mixing}, i.e. the distribution of its state after polynomially many transitions is $\varepsilon$-close to its stationary distribution. A random sample~$\Gamma$ is obtained by running this chain starting with $\emptyset$ for $\mathcal O(\Delta n\log(n/\varepsilon))$ time. Our polymer representation for Holants yields a bijection between the families in $\mathcal I(K)$ and assignments $\sigma$ of the edges, thus a sampling algorithm for polymers implies a sampling algorithm for Holants.

Chen et al. studied a Markov chain for polymers representing a vertex spin system and established conditions to efficiently sample using this chain~\cite{2019:Chen:fast_algorithms_low_temperature}. We obtain our algorithms by adapting their approach to polymer models that originate from Holants. The sampling algorithms give fast randomised approximate counting algorithms. We discuss the technical details in Section~\ref{sec:random}.

\subsection{Related literature and discussion of our results}

\paragraph{Holant polynomials}

Most Holant polynomials considered in the literature study special cases of graph polynomials. Among them, the matching polynomial $Z_{\Matchings}(G,\pi,(z,1))$ of a graph $G$ is perhaps the most studied from an algorithmic perspective. It was first studied in statistical physics as the partition function of the monomer-dimer model~\cite{1972:Heilmann:Monomer_Dimer_systems}. The first approximation algorithm for the matching polynomial was the MCMC algorithm of Jerrum and Sinclair~\cite{1989:Jerrum:approximating_permanent}. Barvinok~\cite{2016:Barvinok:computing_permanent_complex_matrices} (see also \cite[Section~5.1]{2016:Barvinok:combinatorics_complexity_partition_functions}) was the first to connect absence of zeros with approximation algorithms. He gave a quasi-polynomial algorithm for all $z\in\C$ with $|z|\leq (1/(4\Delta-1))$ by computing $\mathcal{O}(\log |V(G)|)$ coefficients of the Taylor series expansion of the logarithm of the matching polynomial. Patel and Regts~\cite{2017:Patel:poly_time_approx_partition_function} refined the computation of coefficients and gave an FPTAS for this region. On the other hand, Bez{\'{a}}kov{\'{a}} et al. \cite{2018:Bezakova:complexity_approx_matching_polynomial} showed that it is \nump{}-hard to approximate $Z_{\Matchings}(G,(1,z))$ when $z\in \mathbb R$ and $z<-1/(4(\Delta-1))$. 

There is an inherent connection between the algorithmic results and the location of the roots of the matching polynomial. It is known that all roots are negative reals with $z<-1/(4\Delta -1)$~\cite{1972:Heilmann:Monomer_Dimer_systems}. Directly applied to the matching polynomial Theorem~\ref{thm:Holant_poly_FPTAS} yields an FPTAS for the region $|z|<1/(\Delta2\ex^2)$. However, adjusting the calculations to exploit the structure of matchings we obtain an improved bound of $|z|<1/(\ex(2\Delta -1))$ (see Section~\ref{sec:intro-Boolean}).

Another graph polynomial studied in the literature is the edge cover polynomial. Translated to our framework, this is the  Holant polynomial $Z_{\mathcal C}(G,\pi,(1,z))$ where $\mathcal C$ contains the functions that evaluate to 1 if at least one of the inputs is 1. Liu, Lu and Zhang~\cite{2014:Liu:FPTAS_weighted_edge_covers} discovered an FPTAS for the edge cover polynomial for $z\in\R_{\geq0}$. Csikv\'{a}ri and Oboudi~\cite{2011:Csikvari:roots_edge_cover_polynomial} showed that the roots of the edge cover polynomial - contrary to the matching polynomial - can take imaginary values and are contained in  $\{z\in\C\mid |z|\leq 5.1\}$.

A different kind of Holant polynomial was studied by Lu, Wang and Zhang~\cite{2014:Lu:FPTAS_weighted_fibonacci_gates}. In their setting each edge contributes its own individual weight instead of each domain element as in our case. Their results include approximation algorithms for real weighted Holants with a special type of signatures called Fibonacci gates.

\paragraph{Holant problems}

There is an assiduous ongoing effort to characterise the computational complexity of Holant problems. The literature on exact computations of Holant problems is extensive and most results restrict the signatures to be of Boolean domain and symmetric, i.e. their value only depends on the Hamming weight of their input~\cite{2016:Cai:dichotomy_vanishing_signatures}. Due to the complexity of the problem only few results go beyond symmetric signatures~\cite{2018:Cai:dichotomy_real_Holant,2017:Lin:complexity_Holant_Boolean_weights,2018:Backens:dichotomy_complex_Holant} or consider higher domain symmetric functions \cite{2013:Cai:dichotomy_Holant_domain_3,2014:Cai:complexity_edge_colorings_higher_domain_Holant}.

When considering the complexity of approximating Holant problems there are classical results targeting particular cases such as matchings~\cite{1989:Jerrum:approximating_permanent}, weighted even subgraphs~\cite{1993:Jerrum:poly_time_approximation_Ising} and edge covers~\cite{1997:Bubley:graph_orientations_no_sink} and newer results on restricted classes of Holant problems~\cite{2013:McQuillan:approximating_Holant,2016:Huang:canonical_paths_art_science,2019:Cai:approximability_six_vertex_model}. Recently, Guo et al.~\cite{2019:Guo:Holant} gave a complexity characterisation for a subclass of Boolean symmetric functions called ``generalised second order recurrences''. Their results rely on proving zero-free regions and only apply to Holant problems where each vertex of the input graph is mapped to the same signature. They remark that it is not clear how to get approximation algorithms when the vertices are mapped to different signatures. Theorem~\ref{thm:Holant_FPTAS} partially addresses this as it allows for mixed signature classes.

To our knowledge, the only result on zero-free regions of Holant problems with arbitrary domain size that includes non-symmetric signatures and the possibility to assign mixed signatures to vertices is due to Regts~\cite{2018:Regts:Zero-free_regions}. As discussed in \cite{2017:Patel:poly_time_approx_partition_function} this result can be used to derive FPTAS' for such Holant problems. The results in~\cite{2018:Regts:Zero-free_regions} require the signatures to output a complex value close to~1 on any input which makes them incomparable to Theorem~\ref{thm:Holant_FPTAS}, requiring our signatures to output only complex values close to 0. 
An advantage of Theorem~\ref{thm:Holant_FPTAS} is that it allows for signatures to encode hard constraints, i.e. to take the value 0. 

As an interesting side note, we mention possible implications for the problem of counting perfect matchings, a central problem in computational counting whose complexity remains unresolved. There is an expanding list of approximation problems that are equivalent to counting perfect matchings \cite{2013:McQuillan:approximating_Holant,2019:Guo:Holant,2019:Cai:Pefect_Matchings}. In one of these results, Guo et al.~\cite{2019:Guo:Holant} show approximation-equivalence between counting perfect matchings and some classes of Holant problems. These signature families are not included in Theorem~\ref{thm:Holant_FPTAS} for the Holant problems but are captured by Theorem~\ref{thm:Holant_poly_FPTAS}. Thus, if one could show a reduction to a Holant polynomial for a value of $z$ within the ones in Theorem~\ref{thm:Holant_poly_FPTAS}, then one would get an approximation algorithm for counting perfect matchings.

\paragraph{Further applications of our technique}

As we already remarked, Theorem~\ref{thm:polymers_fptas} applies to general polymer systems, that is the polymer system does not need to originate from a graph theoretic problem.
In Section~\ref{sec:outlook} we discuss the potential extensions of our technique to the problem of counting weighted solutions to a system of sparse linear equations. This problem was recently studied by Barvinok and Regts~\cite{2019:Barvinok:integer_points} where they obtained zero-free regions and approximation algorithms. We show how to express this problem naturally as a polymer system and use the Koteck{\'y}--Preiss condition to obtain bounds for zero-free regions and deterministic algorithms. Although our bounds for the general case are weaker than the ones in~\cite{2019:Barvinok:integer_points} we can get improved bounds for some cases. In particular, one can define a univariate polynomial expressing perfect matchings in hypergraphs as deviations from a ground perfect matching. For this polynomial we use our technique and improve the bounds in~\cite{2019:Barvinok:integer_points} (see Section~\ref{sec:hpm} for details). It remains open to see if, by refining the analysis, one can obtain better bounds for the general case.

%% file: boolean.tex
In this section we illustrate our technique on the example
of Boolean Holant polynomials. The results sketched here are subsumed by the general results of Section~\ref{sec:holant_polynomials}. 
Formally, for a graph $G=(V,E)$ and an assignment $\pi\colon V \to \SetFunctions$ the \emph{Boolean Holant polynomial} is defined as
\[
	Z_\SetFunctions(G,\pi,(z_0, z_1))=\sum_{\sigma\in \{0,1\}^E} \prod_{v\in V}\left(f_v(\sigma|_{\adjedges(v)}) z_0^{|\sigma|_0} z_1^{|\sigma|_1}\right),
\]
where for a vertex $v \in V$ the set of its incident edges is denoted by $E(v)$ and for $i\in \{0,1\}$ the number of edges mapped to $i$ by $\sigma$ is denoted by $|\sigma|_i$. We assume without loss of generality $z_0\neq0$ and translate the Boolean Holant polynomial to a univariate polynomial of $z=z_1/z_0$. Furthermore, we assume the empty graph to yield a function value of $1$, i.e.~$f(\mathbf 0)=1$, and denote the class of such Boolean signatures by $\SetFunctions_1=\{f\mid f(\mathbf 0)=1\}$. Therefore, for $\SetFunctions \subseteq \SetFunctions_1$ the Holant polynomial $Z_\SetFunctions(G,\pi,(z_0, z_1))$ translates to $Z_{\SetFunctions}(G,\pi,z)$.\par
%
%
We define the set $\mathcal{C}(G)$ of \emph{polymers} of~$G$ to be all connected subgraphs of~$G$ containing at least one edge. Such general subgraphs without isolated vertices can be seen as \emph{edge-induced} and capture the characteristics of the Holant framework. The basic idea of this definition is that assignments~$\sigma$ can be expressed by their corresponding edge-induced connected subgraphs. 
Two polymers $\gamma,\gamma'\in \mathcal C(G)$ are \emph{compatible} if their vertex-sets are disjoint, otherwise they are \emph{incompatible}, i.e.~$\gamma\not\sim\gamma'$ if and only if $V(\gamma)\cap V(\gamma')\neq\emptyset$ and thus $\not \sim$ is reflexive and symmetric. We denote by $\mathcal{I}(G)$ 
the collection of all finite sets of polymers from $\mathcal{C}(G)$ that are pairwise compatible. By definition, each set $\Gamma\in\mathcal{I}(G)$ corresponds to a collection of vertex-disjoint connected subgraphs of~$G$. Hence, we interpret $\Gamma$ as a graph and write $V(\Gamma)$ and $E(\Gamma)$ to refer to the vertices and edges of $\Gamma$, respectively. 


For a mapping $\pi\colon V\rightarrow \mathcal F_0$ we define a corresponding family of polymer functionals (weights associated to the polymers) $\Phi_\pi(\cdot,z)\colon \mathcal C(G)\rightarrow \mathbb C$  with $z\in\C$   by
\[
	\Phi_\pi(\gamma, z)=z^{|E(\gamma)|}\prod_{v\in V(\gamma)} f_v(\mathds{1}_{E(\gamma)}),
\] 
where $\mathds{1}_A$ denotes the function which assigns~$1$ to every entry that corresponds to the elements of $A$ and~$0$ otherwise. 
The partition function for this graph polymer model then translates to
\[
	Z(\mathcal C(G),\Phi_\pi(\cdot,z)) = \sum_{\Gamma\in \mathcal{I}(G)} \prod_{\gamma\in\Gamma} \Phi_\pi(\gamma, z) = \sum_{\Gamma\in \mathcal{I}(G)} \prod_{\gamma\in\Gamma} 
	z^{|E(\gamma)|}\prod_{v\in V(\gamma)} f_v(\mathds{1}_{E(\gamma)})\,.
\]

Every assignment~$\sigma$ uniquely corresponds to the set of connected subgraphs induced by the edges $e\in E$ with $\sigma(e)=1$. This yields a bijective mapping from the set of assignments $\sigma\in\{0,1\}^{V}$ to the set of families $\Gamma\in\mathcal I(G)$. We can further observe that the weight contribution of an assignment~$\sigma$ in the Holant polynomial is equal to the weight contribution of its respective family~$\Gamma$ to the polymer partition function. Hence,
\[
	Z_{{\SetFunctions}_0}(G,\pi,z)=Z(\mathcal C(G),\Phi_\pi(\cdot,z)).
\]
One of the huge benefits a representation as a polymer partition function yields is the access to a big repertoire of results from statistical physics on abstract polymer systems. The \emph{cluster expansion}~\cite{friedli_velenik_2017} yields a representation of $\log Z(\mathcal C(G),\Phi_\pi(\cdot,z))$ by sums of weights of connected subgraphs of the polymer graph, but only for values of $z$ in a region where $Z(\mathcal C(G),\Phi_\pi(\cdot,z))$ has no roots. In such a zero-free region it is also known that an additive approximation for the logarithm yields a multiplicative approximation to the original function. Proving that the partition function $Z(\mathcal C(G),\Phi_\pi(\cdot,z))$ has no roots in a certain region is hence central to deriving approximation results. In this regard, the polymer representation enables the application of the following useful result which holds for general polymer models $(K,\not\sim, \Phi)$.
\begin{theorem}[Koteck{\'y} and Preiss~{\cite[Theorem~1]{1986:Kotecky:cluster_expansion_polymer_models}}]
	\label{thm:cluster_expansion_polymer_models}
	If there exists a function $a\colon K\rightarrow [0,\infty)$ such that for all $\gamma \in K$,
	\[
	\sum_{\gamma' \not\sim \gamma} |\Phi(\gamma')| \ex^{a(\gamma')} \leq a(\gamma)\,,
	\]
	where the sum is over all polymers $\gamma'$ incompatible with $\gamma$, then the cluster expansion for $\log Z(K,\Phi)$ converges absolutely and, in particular, $Z(K,\Phi) \neq 0$.
\end{theorem}


We use this theorem to show that $Z(\mathcal C(G),\Phi_\pi(\cdot,z))\neq 0$ for types of signatures and regions of $z$ described by the following notation. For the Boolean case, for each signature $f\in\SetFunctions_1$ with arity $d$ it follows that $r(f)=\max\{f(\mathbf{x})\mid\mathbf{x}\in\{0,1\}^d\setminus \{\mathbf 0\}\}$, since $f(\mathbf{0})=1$. Recall that for a function class $\SetFunctions\subseteq\SetFunctions_1$ we defined $r(\SetFunctions)=\max\{r(f)\mid f\in\SetFunctions\}$.
\begin{theorem}\label{boolean_bound}
	For all graphs $G$ of maximum degree $\Delta$, all finite $\SetFunctions\subseteq\SetFunctions_1$, all $\pi\colon V(G)\rightarrow\SetFunctions$ and $z$ in $\{z\in \mathbb C\mid |z|\leq (\Delta\ex^{2}\maxValueOne(\maxValueOne+1))^{-1}\}$ with $\maxValueOne=\max\{1, r(\SetFunctions)\}$, the Taylor series expansion of $\log Z_{\mathcal F}(G,\pi,z)$ converges absolutely.
\end{theorem}
To show this we also use the following result of Borgs, Chayes, Kahn and Lov{\'{a}}sz.
%
%
\begin{lemma}[Borgs, Chayes, Kahn and Lov{\'{a}}sz~{\cite[Lemma~2.1(b)]{2013:Borgs:convergence_graphs_bounded_degree}}]
	\label{lem:connected_subgraph_bound}
	In any graph $G$ of maximum degree $\Delta$, the number of connected subgraphs of $G$ with $m$ edges containing a fixed vertex is at most $\frac{1}{m+1}\binom{(m+1)\Delta}{m}<\frac{(\ex \Delta)^m}{2}$.
\end{lemma}			
%
%
To illustrate where the bound presented in Theorem~\ref{boolean_bound} originates from, consider a fixed polymer $\gamma \in\mathcal C(G)$.  Each $\gamma'$ incompatible with $\gamma$ contributes to the sum that has to be bounded by $a(\gamma)$ the term 
\[
	|\Phi_\pi(\gamma',z)|\,\ex^{a(\gamma')}  = \Big|z^{|E(\gamma')|}\!\!\!\!\prod_{v\in V(\gamma')}\!\!f_v(\mathds{1}_{E(\gamma')})\Big|\ex^{a(\gamma')}  \leq |z|^{|E(\gamma')|}\,\maxValue(\SetFunctions)^{|V(\gamma')|} \,\ex^{a(\gamma')}\,.
\]
We estimate the number of polymers $\gamma'$ incompatible with $\gamma$ with respect to their number of edges. For each $1\leq i\leq m$ denote by $\mathcal C _{\gamma}(i)$ the set of polymers $\gamma'\in \mathcal C(G)$ with $\gamma \not\sim \gamma'$ and $i=|E(\gamma')|$.
First, we observe that by the definition of the polymers in $\mathcal{C}(G)$ the incompatibility $\gamma \not\sim \gamma'$ implies that $\gamma$ and $\gamma'$ share at least one vertex. Second, each polymer in $\mathcal C(G)$ has to be a connected subgraph of $G$. By Lemma~\ref{lem:connected_subgraph_bound} we conclude that $\mathcal C_{\gamma}(i)$ has cardinality at most $|V(\gamma)|\frac{(\ex\Delta)^{i}}{2}$.

With this estimation at hand we choose the function $a(\gamma')=|E(\gamma')|$, which is non-negative on $\mathcal C(G)$. This choice of $a$ with the bounds on the cardinality of $\mathcal C_{\gamma}(i)$ yields
\[
	\sum_{\gamma' \not\sim \gamma} |\Phi_\pi(\gamma', z)| \ex^{a(\gamma')}
	\leq\sum_{i=1}^{|E|} |\mathcal C_{\gamma}(i)|\,|z|^i\, \maxValueOne^{i+1}\,\ex^i
	\leq\sum_{i=1}^{|E|}|V(\gamma)|\frac{(\ex\Delta)^{i}}{2} |z|^i\, \maxValueOne^{i+1}\,\ex^i.
\]
Choosing $|z|\leq (\Delta\ex^{2}\maxValueOne(\maxValueOne+1))^{-1}$ this estimation gives the condition required in Theorem~\ref{thm:cluster_expansion_polymer_models}.

\begin{remark}
 Theorem~\ref{thm:cluster_expansion_polymer_models} allows the choice of any function $a:K\rightarrow[0,\infty)$. The way we estimate the number of incompatible polymers introduces a factor of $|\gamma|$ on the left-hand side of the inequality of this theorem, for some size function $|\cdot|$. We are interested in bounds on $|z|$ that do not depend on the size of polymers which may reach the size of the instance. For this reason we must choose $a(\gamma)\in \Omega(|\gamma|)$. Furthermore, each polymer is incompatible to itself, hence in our partition functions we have to satisfy $|z|^{|\gamma|}e^{a(\gamma)}\leq a(\gamma)$, which for bounds independent on $|\gamma|$, requires $a(\gamma)\in O(|\gamma|)$. Therefore, for all our bounds in this paper we will consider $a(\gamma)$ linear in $|\gamma|$.
\end{remark}
In the zero-free region of $Z(\mathcal C(G),\Phi_\pi(\cdot,z))$ now provided by Theorem~\ref{boolean_bound} we consider approximating $\log Z(\mathcal C(G),\Phi_\pi(\cdot,z))$ and use the cluster expansion representation to do so. An additive approximation for $\log Z(\mathcal C(G),\Phi_\pi(\cdot,z))$ can be computed by the first $\mathcal O(\log|V(G)|)$ coefficients of its Taylor expansion. By the definition of our weights it is obvious that only polymers with few edges contribute to these small coefficients; recall that the cluster expansion represents the logarithm of the partition function by sums of products of polymer weights. For a graph $G$ of bounded maximal degree $\Delta$  Lemma~\ref{lem:connected_subgraph_bound} yields that there are at most $|V(G)|\frac{(\ex \Delta)^m}{2}$ polymers with $m$ edges. Lemma~\ref{lem:connected_subgraph_bound} also allows to bound the degree of the polymer graph restricted to the polymers with few edges. These properties enable an efficient enumeration of all connected subgraphs of the polymer graph which contribute to the first coefficients of the Taylor expansion of  $\log Z(\mathcal C(G),\Phi_\pi(\cdot,z))$. Overall these ideas yield the following theorem.
\begin{theorem}\label{thm:boolean_holant}
	For all graphs $G$ of maximum degree $\Delta$, all $\SetFunctions\subseteq\SetFunctions_1$, all $\pi\colon V(G)\rightarrow\SetFunctions$ and $z$ in $\{z\in \mathbb C\mid |z|\leq (\Delta\ex^{2}\maxValueOne(\maxValueOne+1))^{-1}\}$ with $\maxValueOne=\max\{1, r(\SetFunctions)\}$, the Boolean Holant polynomial $Z_{\mathcal F}(G,\pi,z)$ admits an FPTAS.
\end{theorem}

Our running example of the matching polynomial is also captured by the above theorem. Directly applied, Theorem~\ref{thm:boolean_holant} yields an FPTAS for $z$ with $|z| < 1/(2\Delta\ex^2)$. However, since every polymer consists of exactly one edge, we can improve the bound for $\mathcal C_\gamma$ to $2\Delta-1$ and obtain an FPTAS for $z$ with $|z|\leq 1/(\ex(2\Delta -1))$.


%% file: polymers.tex

\label{sec:polymer_fptas}

We develop a general tool to derive approximation schemes for polymer partition functions.  
We will use the following definition of approximation.
\begin{definition}[\cite{2017:Patel:poly_time_approx_partition_function}]
	Let $q$ and $\zeta$ be non-zero complex numbers. We call $\zeta$ a \emph{multiplicative $\varepsilon$-approximation to $q$} if $e^{-\varepsilon} \leq |q|/|\zeta|\leq e^\varepsilon$ and if the angle between $\zeta$ and $q$ (as seen as vectors in $\mathbb{C} = \mathbb{R}^2$) is at most $\varepsilon$.
\end{definition}

The following observation follows immediately from the definition of the multiplicative $\varepsilon$-approximation.
\begin{observation}
	\label{obs:epsilon-approximation}
	For any $\varepsilon, c > 0$, if $\zeta$ is a multiplicative $\varepsilon$-approximation to some value $q$, then we can also achieve a multiplicative $\varepsilon$-approximation for $cq$, namely $c\zeta$.
\end{observation}
A \emph{fully polynomial-time approximation scheme (FPTAS)} is an algorithm that, given a problem and some fixed parameter $\varepsilon > 0$, produces a multiplicative $\varepsilon$-approximation of its solution in time polynomial in $n$ and $1/\varepsilon$. 

Recall from Section~\ref{sec:intro_method} that a general polymer system consists of a finite set~$K$ of polymers and an incompatibility relation $\not\sim\ \subseteq K\times K$. For a family of weight functionals $\{\Phi(\cdot,z)\colon K\rightarrow \mathbb C\mid z\in \mathbb C\}$ one can consider the corresponding partition functions as a mapping $Z(K,\Phi(\cdot,\cdot))\colon \mathbb C\rightarrow \mathbb C$, with $z\mapsto Z(K,\Phi(\cdot,z))$;
observe that the polymer representation we used for Boolean Holant polynomials falls into this setting with $\Phi(\cdot,z)=\Phi_\pi(\cdot,z)$. 
For each polymer $\gamma\in K$ its weight can then also be seen as a function $\Phi(\gamma,\cdot)\colon \mathbb C\rightarrow \mathbb C$ with $z\mapsto \Phi(\gamma,z)$. 
%

Recently, Helmuth et al.~\cite{2018:Helmuth:algorithmic_pirogov_sinai_theory} also used a polymer representation to derive FPTAS' for some graph polynomials. Their polymer system is used to model problems characterised by \emph{edge}-constraints and assignments on the vertices of the input graph.
Translated to our notation their system is defined as follows.
A polymer $\vdompoly$ for a graph $G$ is a pair  $(\gamma, \vcolouring)$, where $\gamma\in \mathcal C(G)$ and  $\vcolouring\colon V(\gamma)\rightarrow D$. Denote the collection of these polymers by $\vSetPolymers(G)$.  Further, the incompatibility relation is defined differently. In this edge-constraint system polymers are already incompatible if they affect a common edge, which formally translates to $\vdompoly \not\sim \vdompoly'$ if the graph distance of the respective subgraphs $\gamma$ and $\gamma'$  in $G$ is less than~2. In their model they consider weight-functions over only one variable $z$. Additionally, they assume the weight functions $\Phi(\vdompoly,z)$ under study to be analytic functions of $z$ in a neighbourhood of the origin of the complex plane and that there is an absolute constant $\rho >0$ such that for each $\vdompoly\in \vSetPolymers(G)$ the first non-zero term in the Taylor series expansion of $\Phi(\vdompoly, z)$ around zero is of order $k \geq |\gamma| \rho$; this property is what they refer to as \emph{Assumption 1}. For this polymer model, the conditions to derive an FPTAS for the resulting polymer partition function $\vfunction$ are summarised in the following theorem.
%
\begin{theorem}[\cite{2018:Helmuth:algorithmic_pirogov_sinai_theory}]\label{thm:pigorov}
	Fix $\Delta$ and let $\mathfrak{G}$ be a set of graphs of degree at most $\Delta$. Suppose:
	\begin{enumerate}
		\item There is a constant $c$ so that $\vfunction$ is a polynomial in $z$ of degree at most $c|G|$ for all $G \in \mathfrak{G}$.
		\item The weight functions satisfy Assumption 1 and the Taylor coefficients up to order $m$ of  $\Phi(\vdompoly,z)$ can be computed  in time $\exp (\mathcal{O}(m + \log |G|))$ for each $G\in \mathfrak{G}$ and $\vdompoly \in \vSetPolymers (G)$.
		\item For every connected subgraph $G'$ of every $G \in \mathfrak{G}$, we can list all polymers $\vdompoly  \in \vSetPolymers(G)$ with ${\gamma} = G'$ in time $\exp(\mathcal{O}(|G'|))$.
		\item There exists $\delta >0$ so that for all $|z|< \delta$ and all $G \in \mathfrak{G}$, $\vfunction\neq 0$.
	\end{enumerate}
	Then for every $z$ with $|z|<\delta$, there is an FPTAS for $\vfunction$  for all $G \in \mathfrak{G}$.
\end{theorem}

The proof of this result is based on calculating the contribution of polymers of small size to the logarithm of the partition function. Assumption 1 ensures that only those small polymers contribute to the required first coefficients of the Taylor expansion. Enumerating all polymers up to a fixed size however is not enough to reproduce these coefficients as the polymers contribute within families and an enumeration of all small families  would be too costly. To circumvent this problem the contribution of all small families can be computed by an inclusion-exclusion principle on sets of incompatible polymers. More formally, this yields a representation of $\log\vfunction$ not over families of compatible polymers but instead over connected subgraphs of the polymer graph. An enumeration of trees and labellings with incompatible polymers then gives an efficient way to compute the low coefficients of the Taylor expansion of $\log\vfunction$ via this representation.

The crucial step to consider the more convenient representation of $\log\vfunction$ is the so-called \emph{cluster expansion} for hard-core interactions~\cite{friedli_velenik_2017}, which holds not only for the specific function $\vfunction$ but for general polymer partition functions. In our notation, the cluster expansion of the partition function $Z(K,\Phi(\cdot,z))$ of a polymer system $K$ with weights $\Phi(\cdot,z)$ is the following representation (which is always possible as soon as the Taylor expansion of $\log Z(K,\Phi(\cdot,z))$ around zero converges absolutely):
\begin{equation}
\label{eq:cluster_expansion}
	\log Z(K,\Phi(\cdot,z)) =\sum_{k\geq 1}\frac{1}{k!}\sum_{(\gamma_1,\dots,\gamma_k)}\phi(\gamma_1,\dots,\gamma_k)\prod_{i=1}^{k}\Phi(\gamma_i,z).
\end{equation}

The sum in the above display equation is over ordered tuples of polymers. $\phi$ is the \emph{Ursell} function of such a tuple, with 
\[
\phi(\gamma_1,\dots,\gamma_k)=\sum_{G\in T(\gamma_1,\dots,\gamma_k)} (-1)^{|E(G)|},
\]
where $T(\gamma_1,\dots,\gamma_k)$ denotes the set of all spanning connected subgraphs of the polymer incompatibility graph restricted to the polymers that appear in $(\gamma_1,\dots,\gamma_k)$.
Observe that the Ursell function assigns zero to all polymer sets $S$ which are disconnected in the polymer graph, which yields the aforementioned representation of $\log Z(K,\Phi(\cdot,z))$ only requiring connected subgraphs of the polymer incompatibility graph.

In order to describe the properties required to approximate the partition function of a general polymer model, we use the following notation to describe properties of a polymer system which play a role in the computational complexity of approximating its partition function. For $\gamma\in K$ the notion $|\gamma|$ denotes some fixed size-function $|\cdot|\colon K\rightarrow \mathbb R_+$, which can be chosen suitably similar to the function $a$ in Theorem~\ref{thm:cluster_expansion_polymer_models}. Similarly we use $|K|$ to denote a suitable size (not necessarily the cardinality) for the set of polymers $K$. With these notions we can lift \cite[Theorem 2.2]{2018:Helmuth:algorithmic_pirogov_sinai_theory} to the general definition of a polymer model and derive the following set of conditions to efficiently compute approximations to a polymer partition function.
\begin{theorem}
\label{thm:polymers_fptas}
	Let $(K,\not\sim)$ be a polymer system and  $\{\Phi(\cdot,z)\colon K\rightarrow \mathbb C\mid z\in \mathbb C\}$ a family of weight functions with the following conditions:
	\begin{enumerate}
		\item The function $Z(K,\Phi(\cdot,z))$ is a polynomial of $z$ of degree  $d$.
		\item Each function $\Phi(\gamma, z)$,  $\gamma \in K$,  is analytic in an open neighbourhood around zero such that the first non-zero term in its Taylor series expansion around zero is of order $k \geq |\gamma| \rho$ for some $\rho >0$.
		\item For each $m\in \mathbb N$ and  $\gamma\in K$, the Taylor coefficients up to order $m$ of $\Phi(\gamma,z)$ can be computed in time $\exp (\mathcal{O}(m + \log |K|))$.              
		\item All polymers ${\gamma} \in K$ with $|{\gamma}| \leq m$ can be enumerated in $\exp(\mathcal{O}(m)+\log |K|)$.
		\item For each  polymer $\gamma \in K$  all polymers $\gamma'\in K$ with $\gamma' \not\sim \gamma$ and $|{\gamma'}|\leq m$ can be enumerated in time $\exp(\mathcal{O}(m + \log |{\gamma}|))$.
		
		\item There exists $\delta >0$ so that for all $|z|< \delta$, $\log(Z(K,\Phi(\cdot,\cdot)))$ converges absolutely.
	\end{enumerate}
	Then for every $z$ with $|z|<\delta$ and every $\varepsilon>0$ an $\varepsilon$-approximation for $Z(K,\Phi(\cdot,z))$ can be computed in 
	$\exp(\mathcal O(\frac 1\rho(1-\frac{|z|}{\delta})^{-1}\log(\frac d \varepsilon)+\log(|K|)))$.
\end{theorem}
\begin{proof}
	First assume $Z(K,\Phi(\cdot,z))$ is a polynomial of $z$ degree $d$. Absolute convergence of $\log Z(K,\Phi(\cdot,z))$ for all $z$ with $|z|<\delta$ implies that $Z(K,\Phi(\cdot,z))$ has no roots in the open disc of radius $\delta$ around zero. By \cite[Lemma 2.2]{2017:Patel:poly_time_approx_partition_function} exponentiating the value of the Taylor series of $\log Z(K,\Phi(\cdot,z))$ truncated at $m=(1-\frac{|z|}{\delta})^{-1}\log(\frac d \varepsilon)$ yields a multiplicative $\varepsilon$-approximation for $Z(K,\Phi(\cdot,z))$, for each $z$ with $|z|<\delta$.
	
	Due to the representation given by \eqref{eq:cluster_expansion} and the absolute convergence of this series the $k$-th Taylor coefficient of $\log Z(K,\Phi(\cdot,z))$ can be computed with the Taylor coefficients of order $k$ of the  products of functions $\Phi(\gamma,z)$ for $\gamma\in S \subseteq K$ (here the subsets $S$ contain the elements of the corresponding $k$-tuple in \eqref{eq:cluster_expansion}). By condition~2, the $k$-th Taylor coefficient of $\Phi(\gamma,z)$ is non-zero only for polymers of size at most $\frac k\rho$. All $S\subseteq K$ which play a role in computing the $k$-th Taylor coefficient of $\log Z(K,\Phi(\cdot,z))$ are connected subgraphs of the polymer graph (otherwise the Ursell function is zero) which only contain at most  $k$ polymers of size at most $\frac k\rho$.
	
	Enumerating all these sets $S$ can be done like in the proof of~\cite[Theorem 2.2]{2018:Helmuth:algorithmic_pirogov_sinai_theory} by enumerating all trees with at most $k$ vertices (to be spanning trees that ensure connectivity) and assigning them polymers of size at most $\frac k \rho$ such that only incompatible polymers are assigned to adjacent vertices to create all subsets $S\subseteq K$ such that $G(S)$ is connected. Property~3 ensures that it is possible to efficiently enumerate the at most $\exp(\mathcal{O}(k)+\log |K|)$ polymers to choose for the root vertex, and moving downward in the tree, property~4 restricts and efficiently computes the choices to pick an incompatible neighbouring polymer. Regardless of the underlying polymer structure, the Ursell function of a graph $H$ (subgraph $G(S)$ of the polymer graph of the general polymer system) can be computed in $\ex^{\mathcal O(|V(H)|)}$ as described in~\cite[Lemma 2.6]{2018:Helmuth:algorithmic_pirogov_sinai_theory}. This enumeration of the relevant sets $S\subseteq K$ and the computation of their contribution to the $k$-th Taylor coefficient of $\log Z(K,\Phi(\cdot,z))$ can be done in $\exp(\mathcal{O}(k)+\log |K|)$. Overall, an $\varepsilon$-approximation for $Z(K,\Phi(\cdot,\cdot))$ can hence be computed by computing the $k$-th Taylor coefficient of $\log Z(K,\Phi(\cdot,z))$  for all $k\leq m$ which yields a total running time in  $\exp(\mathcal{O}((1-\frac{|z|}{\delta})^{-1}\log(\frac d \varepsilon))+\log |K|)$.
\end{proof}

%% file: holant_as_polymers.tex

In this section we explain how to express a Holant polynomial as the partition function of a polymer system. Recall the definition of a Holant polynomial

\[
Z_\SetFunctions(G,\pi,\mathbf{z})=\sum_{\sigma\in \domain^E} \prod_{v\in V(G)}f_v(\sigma|_{\adjedges(v)})\prod_{i=0}^\domainvar z_i^{|\sigma|_i},
\]
where $|\sigma|_i$ denotes the number of edges for which $\sigma$ assigns the value $i$. We implicitly assume a fixed total order on the edges in $G$ which gives a well-defined mapping of a subset of edges to a vector: in the above definition $\sigma|_{\adjedges(v)}$ represents the vector $(\sigma(e_1),\dots,\sigma(e_d))$, where $(e_1,\dots,e_d)$ is the ordered set of edges adjacent to $v$.


To represent  a general Holant polynomial as a polymer partition function we restrict the class of signatures to $\mathcal F_0=\{f\mid f(\mathbf{0})\neq 0\}$. We further require that $z_0\neq 0$, since if we are interested in approximating Holants with $\mathbf{z}$ that contains $0$ values, we can re-translate the Holant to reduce the domain and exclude these elements as any assignment that maps to them will not contribute to the partition function. These requirements ensure that the assignment $\sigma_0$ that maps every edge to $0$ contributes a non-zero weight to the partition function. The families of compatible polymers will express assignments in terms of their ``deviations'' from the ground state $\sigma_0$. 

Given $Z_{\SetFunctions_0}(G,\pi,\mathbf{z})$ over domain $\domain=\{0,1,\dots,\domainvar\}$ we define the following polymer model. A polymer $\dompoly$ for $G$ is a pair  $(\gamma, \phi_\gamma)$, where $\gamma\in \mathcal C(G)$ (so a polymer as before, a connected subgraph of $G$ with at least one edge) and  $\phi_\gamma\colon E(\gamma)\rightarrow [\domainvar]$. Let $\mathcal C_\domainvar(G)$ denote the set of all such polymers for $G$. As in the Boolean case, for $\dompoly,\dompoly'\in\mathcal C_\domainvar(G)$, we say that two polymer are incompatible~$\dompoly\not\sim\dompoly'$, if $V(\gamma)\cap V(\gamma')\neq\emptyset$.
Finally, we write $\mathcal{I}_\domainvar(G)$ to denote the collection of finite subsets of pairwise compatible polymers of $\mathcal C_\domainvar(G)$.

%

For $\pi\colon V\rightarrow \SetFunctions_0$, we define the weight function $\Phi_\pi$ for each $\dompoly\in \mathcal C_\domainvar(G)$ and $\mathbf{z}\in \mathbb C^{\domainvar+1}$ by 
\[
\Phi_\pi(\dompoly, \mathbf{z})=\left(\prod_{i=1}^\domainvar \left(\frac{z_i}{z_0}\right)^{|\adjkappa(\dompoly)|_i}\right)\prod_{v\in V(\gamma)} f_v(\adjkappa(\dompoly))\frac{1}{f_v(\mathbf{0})}
\]
for each polymer~$\dompoly$, where
we write $\adjkappa(\dompoly)$ to denote $(\varphi_\gamma(e_1),\dots, \varphi_\gamma(e_d))$ with $\adjedges(v)=\{e_1,\dots,e_d\}$ (listed in the implicitly assumed fixed total order on the edges), where $\varphi_\gamma(e)=\phi_\gamma(e)$, if $e\in \gamma$ and $\varphi_\gamma(e)=0$, otherwise. Further, $|\adjkappa(\dompoly)|_i$ denotes the number of occurrences of the value $i$ in $\adjkappa(\dompoly)$.

The partition function for this graph polymer model then translates to:
\begin{align*}
Z(\mathcal C_\domainvar(G),\Phi_\pi(\cdot,\mathbf{z})) &= \sum_{\Gamma\in \mathcal{I}_\domainvar(G)} \prod_{\dompoly\in\Gamma} \Phi_\pi(\dompoly, \mathbf{z}) \\
&= \sum_{\Gamma\in \mathcal{I}_\domainvar(G)} \prod_{\dompoly\in\Gamma} 
\left(\prod_{i=1}^\domainvar \left(\frac{z_i}{z_0}\right)^{|\adjkappa(\dompoly)|_i}\right)\prod_{v\in V(\gamma)} f_v(\adjkappa(\dompoly))\frac{1}{f_v(\mathbf{0})}\,.
\end{align*}

Now we can show that our polymer model encodes the Holant as intended, i.e. an $\varepsilon$-approximation for the polymer partition function yields a $\varepsilon$-approximation for the Holant polynomial.

\begin{lemma}
	\label{lem:Holant_polynomial_to_polymer_partition_function}
	For all finite graphs~$G=(V,E)$, all $\SetFunctions\subseteq\SetFunctions_0$, all $\pi\colon V\rightarrow\mathcal F_0$ and all $\mathbf{z}\in \mathbb C^{\domainvar+1}$ with $z_0\neq 0$,
	\[
	Z_{\mathcal F}(G,\pi,\mathbf{z})=z_0^{|E(G)|}\left( \prod_{v\in V} f_v(\mathbf{0})\right) Z(\mathcal C_\domainvar(G),\Phi_\pi(\cdot,\mathbf{z})).
	\]
	\end{lemma}
	\begin{proof}
		Every assignment~$\sigma$ uniquely corresponds to the set of $\domainvar$-edge-coloured connected subgraphs which are induced by the edges $e\in E$ with $\sigma(e)\neq 0$ together with the colouring given by~$\sigma$. More formally, let $G[\sigma]$ be the subgraph of~$G$ induced by the edges $e$ with $\sigma(e)\neq 0$. By definition, each maximal connected component $\gamma$ in $G[\sigma]$ is a polymer in $\mathcal C(G)$, so $(\gamma,\sigma|_{E(\gamma)})$ is a polymer in $\mathcal C_\domainvar(G)$; observe that by definition, on such a component $\sigma$ only assigns the colours in $[\domainvar]$.	Furthermore, the collection of all maximal connected components in $G[\sigma]$ is a set of pairwise compatible polymers in $\mathcal C(G)$, so their coloured versions are pairwise compatible in $\mathcal C_\domainvar (G)$.
		Consequently, the set of maximal connected components in $G[\sigma]$ together with their colourings assigned by $\sigma$ is a set in $\mathcal I_\domainvar(G)$.\par
		Conversely, every set $\bar\Gamma\in\mathcal I_\domainvar(G)$ uniquely corresponds to the assignment $\sigma_{\bar\Gamma}$ defined by $\sigma_{\bar \Gamma}(e)=\phi_\gamma(e)$ if $e\in E(\gamma)$, for some $\dompoly=(\gamma,\phi_\gamma)\in\bar\Gamma$ (observe that in this case there is only one $\dompoly\in \Gamma$ which covers this edge by the compatibility condition, so this is well-defined) and $\sigma_{\bar \Gamma}(e)=0$ otherwise.\par
		We abuse notation and denote families $\bar\Gamma \in \mathcal I_\kappa(G)$ with ``bar'' to use $\Gamma$ to refer to the corresponding collection of subgraphs $\gamma\in \mathcal C(G)$ with $(\gamma,\phi)\in\bar \Gamma$ for some colouring function $\phi$. Such an uncoloured collection $\Gamma$ can then be seen as a subgraph of $G$.
		%
		These observations yield
		\begin{align*}z_0^{|E(G)|}&\left( \prod_{v\in V} f_v(\mathbf{0})\right)Z(\mathcal C_\kappa(G),\Phi_\pi(\cdot,\mathbf{z}))\\
		=& z_0^{|E(G)|}\left( \prod_{v\in V} f_v(\mathbf{0})\right)\sum_{\bar \Gamma\in \mathcal{I}_\kappa(G)} \prod_{\dompoly\in\bar \Gamma} \left(\prod_{i=1}^\domainvar \left(\frac{z_i}{z_0}\right)^{|\adjkappa(\dompoly)|_i}\right)\prod_{v\in V(\gamma)} f_v(\adjkappa(\dompoly))\frac{1}{f_v(\mathbf{0})}\\[2ex]
		=& z_0^{|E(G)|}\sum_{\bar\Gamma\in \mathcal{I}_\kappa(G)} \left( \prod_{v\in V} f_v(\mathbf{0})\right)\prod_{v\in V(G[\Gamma])}\frac{1}{f_v(\mathbf{0})}
		\prod_{v\in V(G[\Gamma])} f_v(\adjkappa(\bar \Gamma)) \prod_{\dompoly\in\bar\Gamma} \left(\prod_{i=1}^\domainvar \left(\frac{z_i}{z_0}\right)^{|\adjkappa(\dompoly)|_i}\right)\\[2ex]
		=&z_0^{|E(G)|}\sum_{\bar\Gamma\in \mathcal{I}_\kappa(G)}\prod_{v\notin V(G[\Gamma])}f_v(\mathbf{0})\prod_{v\in V(G[\Gamma])}f_v(\adjkappa(\bar\Gamma)) \left(\prod_{i=1}^\domainvar \left(\frac{z_i}{z_0}\right)^{|\adjkappa(\bar\Gamma)|_i}\right)\\[2ex]
		=&z_0^{|E(G)|}\sum_{\bar\Gamma\in \mathcal{I}_\kappa(G)}\prod_{v\in V} f_v(\adjkappa(\bar\Gamma)) \left(\prod_{i=1}^\domainvar \left(\frac{z_i}{z_0}\right)^{|\adjkappa(\bar\Gamma)|_i}\right)\\[2ex]
		=&z_0^{|E(G)|}\sum_{\bar\Gamma\in \mathcal{I}_\kappa(G)}\prod_{v\in V} f_v(\sigma_{\bar \Gamma}|_{\adjedges(v)}) \left(\prod_{i=1}^\domainvar \left(\frac{z_i}{z_0}\right)^{|\sigma_{\bar\Gamma}|_i}\right)\\
		=& z_0^{|E(G)|}\sum_{\sigma\in\domain^E} \prod_{v\in V(G)}f_v(\sigma|_{\adjedges(v)})\prod_{i=1}^\domainvar \left(\frac{z_i}{z_0}\right)^{|\sigma|_i}\\
		=&\sum_{\sigma\in\domain^E} \prod_{v\in V(G)}f_v(\sigma|_{\adjedges(v)})\left(z_0^{|E(G)|-\sum_{i=1}^k|\sigma|_i}\right)\prod_{i=1}^\domainvar z_i^{|\sigma|_i}\\
		=&\sum_{\sigma\in\domain^E} \prod_{v\in V(G)}f_v(\sigma|_{\adjedges(v)})\prod_{i=0}^\domainvar z_i^{|\sigma|_i}\\
		=&Z_\SetFunctions(G,\pi,\mathbf{z})\,.
		\end{align*}
		This concludes the lemma.
		\end{proof}
		
We remark that the choice to use colour $0$ as a ground state is only for convenience. If there exists $c\in D$ with $z_c\neq 0$ and the all $c$ configuration~$\sigma_c$ evaluating a non-zero term in the Holant partition function, then we can use a polymer model whose compatible families express deviations from $\sigma_c$. In fact, as we will see in Section~\ref{sec:outlook} we can define polymer systems as deviations from any ground state assignment $\sigma$ with non-zero weight.

%% file: holant_deterministic.tex
We have already explained how to translate a Holant partition function $Z_{\mathcal F_0}(G,\pi,\mathbf{z})$ to a polymer partition function $Z(\mathcal{C}(G),\Phi_\pi(\cdot,\mathbf{z}))$, for $\mathbf{z}$ with $z_0\neq0$. In this section we will first identify the combinations of signature families and values of $\mathbf{z}$ for which the Taylor expansion of the Holants convergence absolutely via the Koteck{\'y}--Preiss condition and then show how Theorem~\ref{thm:polymers_fptas} applies.

\subsection{Holant polynomials}	\label{sec:polynomials_fptas}

			We begin by establishing absolute convergence for the logarithm of Holant polynomials. 
			As in the Boolean case, we use an upper bound on the number of connected subgraphs of the input graph $G$ that contain a fixed vertex similar to Lemma~\ref{lem:connected_subgraph_bound}. As we are now working in domain of arbitrary size, we will use the following generalisation of Lemma~\ref{lem:connected_subgraph_bound}.
			
		\begin{corollary}\label{lem:connected_subgraph_bound_kappa}
			In any graph $G$ of maximum degree $\Delta$, the number of  connected $\domainvar$-edge-coloured subgraphs of $G$ with $m$ edges containing a fixed vertex  is at most $\frac{(\Delta\domainvar\ex)^m}{2}$.
			\end{corollary}
			\begin{proof}	 
			Considering the $\domainvar^m$ different colourings that can be assigned to a graph with $m$ edges and using Lemma~\ref{lem:connected_subgraph_bound}, the claimed bound follows immediately.
			\end{proof}

Recall that for $f\in\domainvarfunctions$ with arity $d$ we define $r(f)=\max_{x\in\domain^d\setminus\{\mathbf 0\}}\{|f(x)|/|f(\mathbf{0})|\}$ and that for a function class $\SetFunctions\subseteq\domainvarfunctions$ we define $r(\SetFunctions)=\max_{f\in\SetFunctions}\{r(f)\}$. 
			
				\begin{lemma}\label{lem:holant_poly_zeros}
					Let $\mathcal F\subseteq \mathcal F_0$. For all graphs $G$ of maximum degree $\Delta$ and all $\pi\colon V\rightarrow \mathcal F$, the Taylor expansion of $\log Z_{\mathcal F}(G,\pi,\mathbf{z})$ converges absolutely in \[\left\{(z_0,z_1,\dots,z_\domainvar)\in \mathbb C^{\domainvar+1}\ \Big|\ z_0\neq0, \frac{|z_i|}{|z_0|}\leq (\Delta\domainvar\ex^{2}\maxValueOne(\maxValueOne+1))^{-1},\ 1\leq i\leq \domainvar\right\}\]
					with $\maxValueOne=\max\{1, r((\mathcal F))\}$.
					\end{lemma}
					\begin{proof}
We want to apply Theorem~\ref{thm:cluster_expansion_polymer_models} on the polymer representation $Z(\mathcal C_\domainvar(G),\Phi_\pi(\cdot,\mathbf{z}))$. Consider a fixed  $\dompoly \in\mathcal C_\domainvar(G)$. Each polymer $\dompoly'$ incompatible with $\dompoly$ contributes to the sum that has to be bounded by $a(\dompoly)$ the term
\begin{align*}
|\Phi_\pi(\dompoly',\mathbf{z})|\,\ex^{a(\dompoly')}= & \left(\prod_{i=1}^\domainvar \left(\frac{|z_i|}{|z_0|}\right)^{|\adjkappa(\dompoly')|_i}\prod_{v\in V(\gamma')} |f_v(\adjkappa(\dompoly'))|\frac{1}{|f_v(\mathbf{0})|}\right)\ex^{a(\dompoly')}\\ \leq & \ \left(\max_{1\leq i\leq \domainvar}\left\{\frac{|z_i|}{|z_0|}\right\}\right)^{|E(\gamma')|}\,\maxValueOne^{|V(\gamma')|} \,\ex^{a(\dompoly')}.
\end{align*}
 
In the following we will estimate the number of polymers $\dompoly'$ incompatible with $\dompoly$ with respect to their number of edges. To this end, denote by $\mathcal C _\domainvar^{\dompoly}(i)$ the set of polymers $\dompoly'\in \mathcal C_\domainvar(G)$ with $\dompoly' \not\sim \dompoly$ and $i=|E(\gamma')|$, for each $1\leq i\leq |E(G)|$.
First, observe that by the definition of the polymers in $\mathcal{C}_\domainvar(G)$ the incompatibility $\dompoly' \not\sim \dompoly$ implies that the corresponding subgraphs $\gamma'$ and $\gamma$ share at least one vertex. Further,   $\gamma'$ has to be a connected subgraph of $G$. By Lemma~\ref{lem:connected_subgraph_bound_kappa} we conclude that $\mathcal C_\domainvar^{\dompoly}(i)$ has cardinality at most $|V(\gamma)|(\frac{(\Delta\domainvar\ex)^{i}}{2})$.
 
Consider $a(\dompoly)=\alpha|E(\gamma)|$ for some constant $\alpha>0$ to be assigned optimally later. This choice of $a$ with the bounds on the cardinality of $\mathcal C_\domainvar^{\dompoly}(i)$, and the fact that a connected subgraph with $i$ edges contains at most $i+1$ vertices, yields
\begin{align*}
\sum_{\dompoly' \not\sim \dompoly} |\Phi_\pi(\dompoly',\mathbf{z})| \ex^{a(\dompoly')}
\leq & \sum_{i=1}^{|E(G)|} |\mathcal C_\domainvar^{\dompoly}(i)|\,\left(\max_{1\leq i\leq \domainvar}\left\{\frac{|z_i|}{|z_0|}\right\}\right)^i\, \maxValueOne^{i+1}\,\ex^{\alpha i}
\\
\leq &\sum_{i=1}^{|E(G)|}|V(\gamma)|\frac{(\Delta\domainvar\ex)^{i}}{2} \left(\max_{1\leq i\leq \domainvar}\left\{\frac{|z_i|}{|z_0|}\right\}\right)^i\, \maxValueOne^{i+1}\,\ex^{\alpha i}.
\end{align*}
With $\beta$ as shorthand on the upper bound on $\max_{1\leq i\leq \domainvar}\left\{\frac{|z_i|}{|z_0|}\right\}$ to describe the region for $\mathbf{z}$ we can reformulate the condition of Theorem~\ref{thm:cluster_expansion_polymer_models}  to
\begin{align*}
&&\sum_{i=1}^{|E(G)|}& |V(\gamma)|\frac{(\Delta\domainvar\ex)^{i}}{2} \beta^i\, \maxValueOne^{i+1}\,\ex^{\alpha i} &&\leq& \alpha|E(\gamma)| \\
&\Leftrightarrow & \sum_{i=1}^{|E(G)|}&(\Delta\domainvar\ex^{\alpha+1}\beta \maxValueOne)^{i} &&\leq& 2\alpha\frac{|E(\gamma)|}{|V(\gamma)|\maxValueOne} \\
&\Leftarrow & &\frac{1}{1-\Delta\domainvar\ex^{\alpha+1}\beta \maxValueOne} -1 &&\leq& \frac{\alpha}{\maxValueOne} \\
&\Leftrightarrow & & \Delta\domainvar \ex^{\alpha+1}\beta\maxValueOne && \leq & \frac{\alpha}{\alpha + \maxValueOne} \\
&\Leftrightarrow & \beta&&&\leq&\frac{\alpha}{\maxValueOne\Delta\domainvar\ex^{\alpha+1}(\alpha+\maxValueOne)}.
\end{align*}
 
Maximizing the region for $\mathbf{z}$ for which approximation is possible hence means choosing $\alpha$ depending on $\maxValueOne$ to maximize $\frac{\alpha}{\maxValueOne\ex^{\alpha+1}(\alpha+\maxValueOne)}$.\\
The first derivative of this expression with respect to $\alpha$ is
\[\frac{\maxValueOne\ex^{\alpha+1}(\alpha+\maxValueOne)-\alpha(\maxValueOne\ex^{\alpha+1}(\alpha+\maxValueOne)+\maxValueOne\ex^{\alpha+1})}{(\maxValueOne\ex^{\alpha+1}(\alpha+\maxValueOne))^2} \,=\,\frac{\maxValueOne-\maxValueOne\alpha-\alpha^2}{\maxValueOne\ex^{\alpha+1}(\alpha+\maxValueOne)^2}.\]
The only positive value for $\alpha$ that sets this to zero is $\alpha=\frac12(\sqrt{\maxValueOne}\sqrt{\maxValueOne+4}-\maxValueOne)$, and furthermore the second derivative is negative at this point.\\
Plugging this into the bound on $\beta$ gives
\[\beta\leq \frac{\sqrt{\maxValueOne}\sqrt{\maxValueOne+4}-\maxValueOne}{\maxValueOne\Delta\domainvar(\sqrt{\maxValueOne}\sqrt{\maxValueOne+4}+\maxValueOne)\ex^{\frac12(\sqrt{\maxValueOne}\sqrt{\maxValueOne+4}-\maxValueOne)+1}}.\]

Note that any value of $\alpha>0$ gives a bound. As $\maxValueOne$ goes to infinity, $\alpha$ goes to $1$. Setting $\alpha=1$ gives the simpler bound
\[\beta\leq \frac{1}{\Delta\domainvar\ex^{2}\maxValueOne(\maxValueOne+1)}\]
stated in the lemma.
						\end{proof}

What remains to obtain an FPTAS is to show how to apply Theorem~\ref{thm:polymers_fptas} to the polymer partition function $Z(\mathcal C_\domainvar(G),\Phi_\pi(\cdot,\mathbf{z}))$. Observe, that from 
 the definition of the weights $\Phi_\pi(\dompoly, \mathbf{z})$, $Z(\mathcal C_\domainvar(G),\Phi_\pi(\cdot,\mathbf{z}))$ can be a multivariate polynomial, so not necessarily univariate. This introduces a catalogue of issues when attempting to apply Theorem~\ref{thm:polymers_fptas} on $Z(\mathcal C_\domainvar(G),\Phi_\pi(\cdot,\mathbf{z}))$ directly. To circumvent these issues, we introduce an additional variable $x$ and transform the weights $\Phi_\pi(\dompoly, \mathbf{z})$ to univariate polynomials of $x$ where we assume the variables $\mathbf{z}$ to be fixed.

For a given weight $\Phi_\pi(\dompoly, \mathbf{z})$ we define the univariate weight 
\[
\Phi^x_\pi(\dompoly, \mathbf{z},x) = \left(\prod_{i=1}^\domainvar \left(\frac{z_i}{z_0}\right)^{|\adjkappa(\dompoly)|_i}\right) x^{|E(\gamma)|} \prod_{v\in V(\gamma)} f_v(\adjkappa(\dompoly))\frac{1}{f_v(\mathbf{0})} ,
\]
which serves as input for the partition function $Z(\mathcal C_\domainvar(G),\Phi^x_\pi(\cdot,\mathbf{z},x))$. In particular, setting $x=1$ yields the original weight $\Phi^x_\pi(\dompoly, \mathbf{z},1) = \Phi_\pi(\dompoly, \mathbf{z})$ and, interpreted as this univariate polynomial weight-function, the degree of $\Phi^x_\pi(\cdot,\mathbf{z},\cdot)$ corresponds to the number of edges $\dompoly$ which enables the application of Theorem~\ref{thm:polymers_fptas}.

{\renewcommand{\thetheorem}{\ref{thm:Holant_poly_FPTAS}}
\begin{theorem}
	Let $\mathcal F\subseteq \mathcal F_0$. For all graphs $G$ of maximum degree $\Delta$ and all $\pi\colon V\rightarrow \mathcal F$, the Holant polynomial admits an FPTAS for $\mathbf{z}$ in
	\[\left\{(z_0,z_1,\dots,z_\domainvar)\in \mathbb C^{\domainvar+1}\ \Big|\  z_0\neq0, \frac{|z_i|}{|z_0|}< (\Delta\domainvar\ex^{2}\maxValueOne(\maxValueOne+1))^{-1},\ 1\leq i\leq \domainvar\right\}\]
	with $\maxValueOne = \max\{1, r(\mathcal F)\}$.
\end{theorem}
}
\begin{proof}
	We have that due to the definition of $\domainvarfunctions$ and Lemma~\ref{lem:Holant_polynomial_to_polymer_partition_function}
	\[
	Z_{\mathcal F}(G,\pi,\mathbf{z})=z_0^{|E(G)|}\left( \prod_{v\in V} f_v(\mathbf{0})\right) Z(\mathcal C_\domainvar(G),\Phi_\pi(\cdot,\mathbf{z})),
	\]	
	thus it suffices to approximate $Z(\mathcal C_\domainvar(G),\Phi_\pi(\cdot,\mathbf{z}))$. We use the above defined representation as a univariate polynomial $Z(\mathcal C_\domainvar(G),\Phi^x_\pi(\cdot,\mathbf{z},x))$ to apply Theorem~\ref{thm:polymers_fptas} and show how $Z(\mathcal C_\domainvar(G),\Phi^x_\pi(\cdot,\mathbf{z},x))$  satisfies the theorem's conditions. 
	As size functions we choose $|\mathcal C_\domainvar(G)|$ the cardinality of $G$, i.e.~the number of vertices, and as $|\dompoly|$  the number of edges of the support graph. 
	
	\begin{enumerate}
		\item  The partition function $Z(\mathcal C_\domainvar(G),\Phi^x_\pi(\cdot,\mathbf{z},x))$ is a polynomial of $x$ of maximum degree $|E(G)|$.
		\item For each polymer $\dompoly$ its weight $\Phi^x_\pi(\dompoly,\mathbf{z},x)$ is a monomial in $x$ of degree $|E(\gamma)|=|\dompoly|$, thus condition~2 holds with $\rho=1$.
		\item For each polymer $\dompoly$ its weight $\Phi^x_\pi(\dompoly,\mathbf{z},x)$  can be computed exactly, not just its first $m$ Taylor coefficients, which can be done in $\mathcal O(G)=\mathcal O(|\mathcal C_\domainvar(G)|)$. Recall here that we only consider constraint functions that can be computed efficiently for the general ground class $\mathcal F_0$.
		\item The support $\gamma$ of a polymer $\dompoly$ is a connected subgraph of $G$ and adjoined with the edge-colouring $\phi_\gamma$ we obtain, using Lemma~\ref{lem:connected_subgraph_bound_kappa}, that there are at most $\frac{(\Delta\domainvar\ex)^m}{2}$ such edge-coloured connected subgraphs with at most $m$ edges containing a fixed vertex. 
		There are $|V(G)|$ many choices for a fixed vertex and thus there are at most $|V(G)|m{(\Delta\domainvar\ex)^m} \in \exp(\mathcal{O}(m + \log |\mathcal C_\domainvar (G)|))$ polymers of size $m$. We can enumerate these as in~\cite[Lemma~3.7.]{2017:Patel:poly_time_approx_partition_function} by enumerating all connected subgraphs $S \subseteq G$ of size $|S|\leq m$, which requires a run-time in $\mathcal{O}(|V(G)|^2 m^7))$. Afterwards, for every such subgraph~$S$ we add the $\domainvar^{|S|} \leq \domainvar^m$ different edge-colourings and obtain a running time in $\exp(\mathcal{O}(m + \log |V(G)|))$.
		\item For $\dompoly \in \mathcal C_\domainvar (G)$ with $|\gamma|=m$ we have that $\dompoly' \in \mathcal C_\domainvar (G)$ is incompatible with $\dompoly$ if their supports $\gamma$ and $\gamma'$ share at least one vertex $v$. Hence, we enumerate for all $v \in V(\gamma)$ all polymers $\dompoly'$ of size at most $m$ containing $v$ and then remove duplicates. There are at most $\frac{(\Delta\domainvar\ex)^m}{2}$ such edge-coloured connected subgraphs of size $m$ containing a fixed vertex due to Lemma~\ref{lem:connected_subgraph_bound_kappa}. For each $v\in V(\gamma)$, we can then enumerate all connected subgraphs $\gamma'\subseteq G$ containing  $v$ and of size at most $m$ as in~\cite[Lemma~3.4.]{2017:Patel:poly_time_approx_partition_function}, which takes time $\mathcal{O}(|\gamma|^2 m^6)) \subseteq \exp(\mathcal{O}(m + \log |{\gamma}|))$.
		\item From Lemma~\ref{lem:holant_poly_zeros}, we have that $\log (Z(\mathcal C_\domainvar(G),\Phi_\pi(\cdot,\mathbf{z})))$ converges absolutely for $\mathbf z$ with $\frac{|z_i|}{|z_0|}\leq (\Delta\domainvar\ex^{2}\maxValueOne(\maxValueOne+1))^{-1}$ for all $1\leq i\leq \domainvar$. For a fixed $\mathbf z$ in the region defined by the statement of the theorem, denote $q=\min\{\frac{|z_0|}{|z_i|}(\Delta\domainvar\ex^{2}\maxValueOne(\maxValueOne+1))^{-1} \mid 1 \leq i\leq \domainvar\}>1$.	
Adding $x$ with $|x| \leq q$  to the  calculations in  Lemma~\ref{lem:holant_poly_zeros} shows that $\log (Z(\mathcal C_\domainvar(G),\Phi^x_\pi(\cdot,\mathbf{z},x)))$ converges absolutely for  $|x| \leq q$ and $\mathbf z$ in the region defined in the statement of the theorem.
	\end{enumerate}
In summary, the conditions of Theorem~\ref{thm:polymers_fptas} are satisfied with $K=\mathcal C_\domainvar(G)$, $\rho=1$ and $d=E(G)$ and for all $x$ with $|x|\leq q$. 
This shows that a $(1+\varepsilon)$-approximation for $Z(\mathcal C_\domainvar(G),\Phi^x_\pi(\cdot,\mathbf{z},x)$ can be computed in $\exp(\mathcal O((1-\frac{|x|}{q})^{-1}\log(\frac{E(G)}{\varepsilon})+\log(|\mathcal C_\domainvar(G)|)))$. 	With our choice of $|\mathcal C_\domainvar(G)|=|V(G)|$ and with $x$ chosen to be~1, to approximate the original partition function $Z(\mathcal C_\domainvar(G),\Phi_\pi(\cdot,\mathbf{z}))$, this yields an overall running-time in $\exp(\mathcal O((1-\frac{1}{q})^{-1}\log(\frac{E(G)}{\varepsilon})+\log|V(G)|))$ which is in $\exp(\mathcal O(\log(\frac{1}{\varepsilon})+\log|V(G)|))$ for any fixed $\mathbf z$.
	\end{proof}

\subsection{Holant problems}

The partition function of a Holant problem~$Z_\mathcal{F}(G,\pi)$ can be viewed as the value of a Holant polynomial $Z_\mathcal{F}(G,\pi,\mathbf{1})$ at $\mathbf{z}=(1,1,\dots,1)$. This immediately gives a polymer representation for Holant problems. The polymer partition function is simply $Z(\mathcal C_\domainvar(G), \Phi_\pi(\cdot,\mathbf{1}))$. 
 In this section we derive bounds that depend only on the functions of $\mathcal F$ that target Holant problems in particular. This gives the analogue of Lemma~\ref{lem:holant_poly_zeros} for Holant problems.
						
\begin{lemma}
\label{lem:Holant_zeros}
Let $\mathcal F\subseteq \mathcal F_0$ with $\maxValue(\SetFunctions)\leq \max\{(2\sqrt\ex)^{-1}(\Delta\kappa\ex)^{-\frac\Delta2} ,0.2058(\domainvar+1)^{-\Delta}\}$. For all graphs $G$ of maximum degree $\Delta$ and all $\pi\colon V\rightarrow \mathcal F$, the Taylor expansion of $\log Z_{\mathcal F}(G,\pi)$ converges absolutely.
\end{lemma}							
\begin{proof}
We again want to apply Theorem~\ref{thm:cluster_expansion_polymer_models} on the polymer representation $Z(\mathcal C_\domainvar(G),\Phi_\pi(\cdot,\mathbf{1}))$.	Consider a fixed  $\dompoly \in\mathcal C_\domainvar(G)$. With $\mathbf{z}$ fixed to $\mathbf{1}$, each polymer $\dompoly'$ incompatible with $\dompoly$ contributes to the sum that has to be bounded by $a(\dompoly)$ the term
	\begin{align*}
		|\Phi_\pi(\dompoly',\mathbf{1})|\,\ex^{a(\dompoly')}= & \left(\prod_{v\in V(\gamma')} |f_v(\adjkappa(\dompoly'))|\frac{1}{|f_v(\mathbf{0})|}\right)\ex^{a(\dompoly')}\\ \leq & \ \maxValue(\SetFunctions)^{|V(\gamma')|} \,\ex^{a(\dompoly')}.
	\end{align*}
								
	As in the proof of Lemma~\ref{lem:holant_poly_zeros}, we use $\mathcal C_\domainvar^{\dompoly}(i)$ to denote the set of polymers with $i$ edges and incompatible with $\dompoly$, and use the estimate $|\mathcal C_\domainvar^{\dompoly}(i)|\leq |V(\gamma)|(\frac{(\Delta\domainvar\ex)^{i}}{2})$.
								
	Since now $\maxValue(\SetFunctions)<1$, we need to consider a lower bound on the number of vertices in a polymer in $\mathcal C_\domainvar^{\dompoly}(i)$ to upper bound its weight. 
		A connected subgraph of maximum degree $\Delta$ with~$i$ vertices contains at most $\frac{\Delta i}{2}$ edges. This bound together with the choice  $a(\dompoly)=\alpha|V(\gamma)|$ for some $\alpha>0$ to be fixed later and $\maxValue=\maxValue(\SetFunctions)$ yields
	\begin{align*}
		\sum_{\dompoly' \not\sim \dompoly} |\Phi_\pi(\dompoly',\mathbf{1})| \ex^{a(\dompoly')}		\leq & \sum_{i=1}^{|V(G)|} |\mathcal C_\domainvar^{\dompoly}(\tfrac{\Delta i}{2})|\,\maxValue^{i}\,\ex^{i\alpha}
		\\
		\leq &\sum_{i=1}^{|V(G)|}|V(\gamma)|\frac{1}{2}((\Delta\domainvar\ex)^{\frac{\Delta}{2}} \, \maxValue\,\ex^{\alpha})^i.
			\end{align*}
			
With geometric series estimate for this sum, bounding this by  $a(\dompoly)=\alpha|V(\gamma)|$ yields the inequality
	\begin{align*}
|V(\gamma)|\frac 12\left(\frac{1}{1-(\Delta\domainvar\ex)^{\frac{\Delta}{2}} \, \maxValue\,\ex^{\alpha}}-1\right) &\leq& \alpha |V(\gamma)|\\
{1-(\Delta\domainvar\ex)^{\frac{\Delta}{2}} \, \maxValue\,\ex^{\alpha}} &\geq& \frac{1}{2\alpha +1}\\
(\Delta\domainvar\ex)^{\frac{\Delta}{2}} \, \maxValue\,\ex^{\alpha} &\leq& {1-\frac{1}{2\alpha +1}} \\
\maxValue &\leq& \frac{2\alpha}{(2\alpha +1)(\Delta\domainvar\ex)^{\frac{\Delta}{2}} \,\ex^{\alpha}}. \\
	\end{align*}
The best bound for $\alpha$ maximizes the function $\frac{2\alpha}{(2\alpha +1)\ex^{\alpha}}$. The first derivative for $\alpha$ is $\frac{-4\alpha^2-2\alpha +2}{(2\alpha +1)^2\ex^{\alpha}}$.
		This expression is zero when $\alpha=\frac 12$ and the second derivative is negative at this point.
		This yields the bound
		
		\[\maxValue\leq \frac{1}{2(\Delta\domainvar\ex)^{\frac{\Delta}{2}} \,\sqrt{\ex}}. \]

%
	

	Alternatively,  we could use $\mathcal V_\domainvar^{\dompoly}(j)$ to denote the set of polymers with $j$ \emph{vertices} that are incompatible with $\dompoly$. By~\cite{2013:Borgs:convergence_graphs_bounded_degree}, we can use the estimate $|\mathcal V_\domainvar^{\dompoly}(j)|\leq |V(\gamma)|(\domainvar+1)^{j\Delta}$. This bound together with the choice  $a(\dompoly)=\alpha|V(\gamma)|$ for some $\alpha>0$ yields the bound
	\begin{align*}
		\sum_{\dompoly' \not\sim \dompoly} |\Phi_\pi(\dompoly',\mathbf{1})| \ex^{a(\dompoly')}
		\leq & \sum_{j=1}^{|V(G)|} |\mathcal V_\domainvar^{\dompoly}(j)|\, \maxValue^{j}\,\ex^{\alpha j}
		\\
		\leq &\sum_{j=1}^{|V(G)|} |V(\gamma)|(\domainvar+1)^{j\Delta}\, \maxValue^{j}\,\ex^{\alpha j}.
	\end{align*} 
The Koteck{\'y}--Preiss condition is hence satisfied for $\alpha$ and $\maxValue$ such that the inequality
		\begin{align*}
&&\sum_{j=1}^{|V(G)|} |V(\gamma)|(\domainvar+1)^{j\Delta}\, \maxValue^{j}\,\ex^{\alpha j}&\leq \alpha|V(\gamma)|\\
\Leftrightarrow &&\sum_{j=1}^{|V(G)|} ((\domainvar+1)^{\Delta} \maxValue\,\ex^{\alpha })^j&\leq \alpha\\
\Leftarrow && \frac{1}{1-(\domainvar+1)^{\Delta} \maxValue\,\ex^{\alpha }}&\leq \alpha+1\\
\Leftrightarrow && {1-(\domainvar+1)^{\Delta} \maxValue\,\ex^{\alpha }}&\geq \frac{1}{\alpha+1}\\
\Leftrightarrow && {(\domainvar+1)^{\Delta} \maxValue\,\ex^{\alpha }}&\leq 1-\frac{1}{\alpha+1}\\
\Leftrightarrow && \maxValue &\leq \frac{\alpha}{(\alpha+1)(\domainvar+1)^{\Delta} \ex^{\alpha}}\\
	\end{align*} 
holds. To find the best $\alpha$ to maximize the function $\frac{\alpha}{(\alpha+1)\ex^{\alpha}}$, we again find the positive root of the derivative
%
%

This gives $\alpha=\frac{\sqrt 5 -1}{2}$, and the second derivative is negative at this point.
We obtain the bound 
	\[\maxValue\leq  \frac{{\sqrt 5 -1}}{(\sqrt 5 +1)(\domainvar+1)^{\Delta} \ex^{\frac{\sqrt 5 -1}{2}}} \approx 0.2058(\domainvar+1)^{-\Delta}. \]

Together the two bounds give the lemma.
%
\end{proof}

As in the case for Holant polynomials, the above lemma gives rise to FPTAS' for the following Holant problems.
{\renewcommand{\thetheorem}{\ref{thm:Holant_FPTAS}}
\begin{theorem}
 Let $\mathcal F\subseteq \mathcal F_0$, with $\maxValue(\SetFunctions)< \max\{(2\sqrt\ex)^{-1}(\Delta\kappa\ex)^{-\frac\Delta2} ,0.2058(\domainvar+1)^{-\Delta}\}$. For all graphs $G$ of maximum degree $\Delta$ and all $\pi\colon V\rightarrow \mathcal F$, the Holant problem $Z_{\mathcal F}(G,\pi)$ admits an FPTAS.
\end{theorem}
}
\begin{proof}
 Since a Holant problem is a Holant polynomial at $\mathbf{z}=\mathbf{1}$, the proof is identical to the proof of Theorem~\ref{thm:Holant_poly_FPTAS} by using Lemma~\ref{lem:Holant_zeros} instead of Lemma~\ref{lem:holant_poly_zeros}. The only difference is that $q$ is now defined by $(\max\{\frac{(2\sqrt\ex)^{-1}(\Delta\kappa\ex)^{-\frac\Delta2}}{r(\mathcal F)} ,\frac{0.2058(\domainvar+1)^{-\Delta}}{r(\mathcal F)}\})^{\frac 1\Delta}$, where the exponent $\frac 1\Delta$ ensures that in the calculations of Lemma~\ref{lem:Holant_zeros} the term $r^{|V(\gamma')|}$ cancels out the $|x|^{|E(\gamma')|}$ term.
\end{proof}

%% file: random.tex

Recently, Chen et al.~\cite{2019:Chen:fast_algorithms_low_temperature} studied Markov chains on vertex spin polymer models. Their results establish conditions yielding 
fast randomised sampling and counting algorithms, with polynomial run-time and linear dependency on the maximum degree~$\Delta$ of the input graph. We follow their approach and show how their algorithms can adapted to Holant problems and Holant polynomials exploiting the respective polymer models.

Consider a Holant polynomial $Z_\mathcal{F}(G,\pi,\mathbf{z})$, where for the rest of this section, we restrict the values of the signatures in $\mathcal F$ to be non-negative reals and we further restrict  $\mathbf{z}\in(\R_{\geq 0})^{\domainvar+1}$. We want to design fast algorithms that sample from the distribution $\mu_G$ with
\[
	\mu_G(\sigma)=\frac{\prod_{v\in V(G)}f_v(\sigma|_{\adjedges(v)})\prod_{i=0}^\kappa z_i^{|\sigma|_i}}{Z_{\mathcal F}(G,\pi,\mathbf{z})}\,.
\]
To do so, we will use polymer models. Recall our definition of the polymer partition function $Z(\mathcal C_\domainvar(G),\Phi_\pi(\cdot,\mathbf{z}))$ from Section~\ref{sec:holant_polynomials}. The restrictions to the non-negative reals imply $\Phi_\pi(\cdot,\mathbf{z})\in\R_{\geq 0}$. We will design a Markov chain with  the families of $\mathcal I_\domainvar(G)$ as state space with stationary distribution $\mu_{\mathcal C_\domainvar(G)}$, where
\[\mu_{\mathcal C_\domainvar(G)}(\Gamma)=\frac{\sum_{\dompoly\in\Gamma}\Phi_\pi(\dompoly,\mathbf{z})}{Z(\mathcal C_\domainvar(G),\Phi_\pi(\cdot,\mathbf{z}))}\,.\]
To ensure that the Markov chain we  study converges to its stationary distribution in polynomial time we require the following  condition. 

\begin{definition}\label{cond:mixing}
	Let $G$ be a graph. A polymer model $(\mathcal C_\domainvar(G),\Phi_\pi(\cdot,\mathbf{z}))$ satisfies the \emph{mixing condition} if there exists a constant $\xi\in(0,1)$ such that for each $\dompoly\in\mathcal C_\domainvar(G)$,
	\begin{equation*}
	\sum_{\dompoly'\not\sim\dompoly}|E(\dompoly')|\Phi_\pi(\dompoly',\mathbf{z})\leq\xi|E(\dompoly)|\,.
	\end{equation*}
\end{definition}

For every $e_0\in E(G)$, let $\mathcal C_\domainvar(e_0)=\{\dompoly\in \mathcal C_\domainvar(G)\mid e_0\in \gamma\}$ be the set of polymers whose underlying graph contains the edge $e_0$ and let $\alpha(e_0)=\sum_{\gamma\in\mathcal C_\domainvar(e_0)}\Phi_\pi(\dompoly,\mathbf{z})$.  We define the probability distribution $\mu_0$ on $\mathcal C(e_0)\cup\{\emptyset\}$ by $\mu_0(\dompoly)=\Phi_\pi(\dompoly,\mathbf{z})$ and $\mu_0(\emptyset)=1-\alpha(\gamma)$. Definition~\ref{cond:mixing} ensures that $\alpha(e_0)<1$, for every $e_0\in E(G)$ (by choosing any $\dompoly$ with $\gamma=e_0$), hence $\mu_0$ is a valid probability distribution. We are now ready to define the transition of the Holant polymer Markov chain $(\Gamma_t)_{t\in\N}$.
\begin{itemize}
	\item Choose $e_0\in E(G)$ uniformly at random.
	\item If there exists $\dompoly\in \Gamma_t$ with $e_0\in E(\dompoly)$, then $\Gamma_{t+1}\leftarrow \Gamma_t\setminus\{\dompoly\}$ with probability $1/2$.
	\item Otherwise, sample $\dompoly$ from $\mu_0$ and if $X_t\cup\{\dompoly\}\in\mathcal I(G)$, then $\Gamma_{t+1}\leftarrow \Gamma_t\cup\{\dompoly\}$ with probability $1/2$.
\end{itemize}
In all other cases, we set $\Gamma_{t+1}\leftarrow \Gamma_t$.

Let $P$ be the transition matrix of the above chain and let $T(\varepsilon)=\max_{\Gamma\in\mathcal I_\domainvar(G)}\min\{t\mid\|P^t(\Gamma,\cdot)-\mu_{\mathcal C_\domainvar(G)}(\cdot)\|_TV\leq\varepsilon\}$ be its mixing time. We have the following lemma.

\begin{lemma}
	Let $\mathcal F\subseteq \mathcal F_0$ and $G$ be a graph of maximum degree $\Delta$ and let $\pi\colon V(G)\rightarrow\mathcal F$. When the mixing condition applies (Definition~\ref{cond:mixing}), the Holant polymer Markov chain has stationary distribution $\mu_{\mathcal C_\domainvar(G)}$ and mixes in $T(\varepsilon)\in O(\Delta n\log(n/\varepsilon))$ many iterations.
\end{lemma}
\begin{proof}
	It is easy to see that the Markov chain~$(\Gamma_t)_{t\in\N}$  is ergodic and thus it has a unique stationary distribution, which is the limit distribution. We fist show that $\mu_{\mathcal C_\domainvar(G)}$ is the unique stationary distribution of the above Markov chain. Let $\Gamma,\Gamma'\in\mathcal I_\domainvar (G)$ such that there exists $\dompoly\in \mathcal C_\domainvar(G)$ with $\Gamma\cup\{\dompoly\}=\Gamma'$. It suffices to show that $\mu_{\mathcal C_\domainvar(G)}(\Gamma)P(\Gamma,\Gamma')=\mu_{\mathcal C_\domainvar(G)}(\Gamma')P(\Gamma',\Gamma)$. The latter holds, since
	\begin{align*}
	\mu_{\mathcal C_\domainvar(G)}(\Gamma)P(\Gamma,\Gamma') & = \frac{\prod_{\dompoly'\in\Gamma}\Phi_\pi(\dompoly',\mathbf{z})}{Z(G)}\frac{|E(\gamma)|}{|E(G)|}\frac{1}{2}\Phi_\pi(\dompoly) \\
	& =\frac{\prod_{\dompoly'\in\Gamma'}\Phi_\pi(\dompoly',\mathbf{z})}{Z(G)}\frac{|E(\gamma)|}{|E(G)|}\frac{1}{2} = \mu_{\mathcal C_\domainvar(G)}(\Gamma')P(\Gamma',\Gamma).
	\end{align*}
	
	To upper bound the mixing time of this Markov chain we will use path coupling as in~\cite{1999:Dyer:randomwalks_combinatorial_objects}. To this end, we define a metric $\dist(\cdot,\cdot)$ on $\Omega$ by setting $\dist(\Gamma,\Gamma')=1$ if for some $\dompoly\in\mathcal C_\domainvar(G)$, $\Gamma\cup\{\dompoly\}=\Gamma'$ or $\Gamma'\cup\{\dompoly\}=\Gamma$. This function $d$ can be naturally extended to all pairs $\Gamma,\Gamma'$ in $\Omega$ by shortest paths. Formally let $\Gamma=\Gamma_0,\Gamma_1,\dots,\Gamma_\ell=\Gamma'$ be a shortest path between $\Gamma$ and $\Gamma'$ in $\Omega$, then $\dist(\Gamma,\Gamma')=\sum_{i=1}^{\ell} \dist(\Gamma_{i-1},\Gamma)$. We can now observe that the diameter~$W=\max_{(\Gamma,\Gamma')}\{\dist(\Gamma,\Gamma')\}$ of $\Omega$ is at most $n$, since no family in $\mathcal I_\domainvar(G)$ can contain more than $\frac{n}{2}$ polymers.
	
	We define our coupling to be the chain denoted by $(X_t,Y_t)$ as follows. First observe that as we are using path coupling, we only need to define how the chain progresses for pairs  that only differ in one polymer. Let $\dompoly=(\gamma,\Phi)\in\mathcal C_\domainvar(G)$ be such that $X_t\cup\{\dompoly\}=Y_t$, then  $(X_t,Y_t)$ progresses as follows.
	\begin{itemize}
		\item For both $X_t$ and $Y_t$ choose the same $e_0\in E(G)$ uniformly at random.
		\item If $e_0\in\gamma$, let $X_{t+1}\leftarrow X_t$ and $Y_{t+1}\leftarrow Y_t\setminus\{\dompoly\}$ with probability $1/2$. With remaining probability $1/2$, sample $\dompoly'$ from $\mu_0$ and let $X_{t+1}\leftarrow X_t\cup\{\dompoly'\}$ and $Y_{t+1}\leftarrow Y_t$.
		\item If $e_0\notin\gamma$, then $X_t$ behaves as the original chain~$\Gamma_t$ and $Y_t$ copies $X_t$ if possible, otherwise $Y_t$ remains at the same state.
	\end{itemize}

	We will now show that the distance of adjacent states reduces in expectation. We have
	\[\E[\dist(X_{t+1},Y_{t+1})\mid X_t,Y_t]\leq1+\frac{1}{2|E(G)|}\left(-|E(\gamma_0)| + \sum_{\dompoly'\not\sim\dompoly}|E(\gamma')|\Phi_\pi(\dompoly',\mathbf{z})\right),\]
	since the distance decreases by $1$ with probability $1/2$ if we choose an edge in $\gamma$ and increases by $1$ with probability $1/2$ if we sample a polymer that is incompatible with $\dompoly$. The latter occurs with probability $|E(G)|^{-1}\sum_{\dompoly'\not\sim\dompoly}|E(\gamma')|\Phi_\pi(\dompoly',\mathbf{z})$.
	From the statement of the lemma the mixing condition applies, which yields
	\[\E[\dist(X_{t+1},Y_{t+1})\mid X_t,Y_t]\leq1-|E(\gamma)|\frac{1-\xi}{2|E(G)|}\leq1-\frac{1-\xi}{2|E(G)|}.\]
	
	Now we can use the path coupling lemma~\cite[(17) in Section~6]{1999:Dyer:randomwalks_combinatorial_objects} which yields $T(\varepsilon)\leq\log(W/\varepsilon)2|E(G)|/(1-\xi)\in O(\Delta n\log(n/\varepsilon))$.
\end{proof}
To obtain a fast sampling algorithm, we have to show that each iteration of the polymer Markov chain only requires expected constant time. To do so, as in Chen et al.~\cite[Definition~4]{2019:Chen:fast_algorithms_low_temperature}, we need to identify under which condition this holds.

\begin{definition}\label{cond:sampling}
	Given a graph $G$ of maximum degree $\Delta$ we say that a polymer model $(\mathcal C_\domainvar(G),\Phi(\cdot,\mathbf{z}))$ satisfies the polymer \emph{sampling condition} with constant $\tau\geq5+3\ln(\kappa\Delta)$ if $\Phi(\dompoly,\mathbf{z})\leq \ex^{-\tau|E(\gamma)|}$ for all $\dompoly\in\mathcal C_\domainvar(G)$.
\end{definition}

Mimicking~\cite{2019:Chen:fast_algorithms_low_temperature} we consider the following algorithm to sample from $\mu_0$ in expected constant time. Let $\varrho=\tau-2-\ln(\kappa\Delta)\geq3+2\ln(\kappa\Delta)$.
\begin{itemize}
	\item Choose $k$ according to the geometric distribution with parameter $1-\ex^{-\varrho}$, that is $\P(k=i)=(1-\ex^{-\varrho})\ex^{-\varrho i}$. Note that $\P(k\geq i)=\ex^{-\varrho i}$.
	\item List all polymers in $\mathcal C_\domainvar(e_0,k)=\{\dompoly\in\mathcal C_\domainvar(e_0)\mid |E(\gamma)|\leq k\}$ and compute their weight functions $\Phi(\cdot,\mathbf{z})$. Recall that $\mathbf{z}$ is considered to be fixed.
	\item Mutually exclusively output $\dompoly\in\mathcal C_\domainvar(e_0,k)$ with probability $\Phi(\dompoly,\mathbf{z})\ex^{\varrho|E(\gamma)|}$ and with all remaining probability output $\emptyset$. Observe that if $k=0$ then we output $\emptyset$ with probability~1.
\end{itemize}

\begin{lemma}
	The above sampling algorithm under the polymer sampling condition (Definition~\ref{cond:sampling}) outputs a polymer $\gamma$ from the distribution $\mu_e$ in expected constant time.
\end{lemma}
\begin{proof}
	 The proof follows the proof of \cite[Lemma~16]{2019:Chen:fast_algorithms_low_temperature}. We first show that the algorithm is well defined, i.e., the probabilities $\Phi(\dompoly,\mathbf{z})\ex^{\varrho|E(\gamma)|}$ sum to less than 1. Using Lemma~\ref{lem:connected_subgraph_bound} we have
	\begin{align*}
	\sum_{\gamma\in\mathcal C_\domainvar(e_0)}\Phi(\dompoly,\mathbf{z})\ex^{\varrho|E(\gamma)|}&\leq \frac{1}{2}\sum_{k\geq 1} \ex^{\tau k-\varrho k}(e\kappa\Delta)^k\\
	&= \frac{1}{2}\sum_{k\geq 1}\ex^{-k(2+\ln(\kappa\Delta))} (\ex\kappa\Delta)^k\\
	&= \frac{1}{2}\sum_{k\geq 1}\ex^{-k}< 1.
	\end{align*}
	
	Next, we will show that the output of the sampler has distribution $\mu_0$. Let $\dompoly\in\mathcal C_\domainvar(e_0)$. In order to have non-zero probability to output $\dompoly$ we must choose $k\geq|E(\gamma)|$. This happens with probability $\ex^{-\varrho|E(\gamma)|}$ by the distribution of $k$. Conditioned we choose an appropriate $k$, the probability of the sampler to output $\dompoly$ is then $\Phi(\dompoly,\mathbf{z})\ex^{\varrho|E(\gamma)|}$. Thus, the probability of choosing any $\dompoly\in\mathcal C_\domainvar(e_0)$ is $\Phi(\dompoly,\mathbf{z})$, showing that the outputs of our sampler are distributed according to $\mu_0$.
	
	Finally we analyse the expected run-time of the sampler. Note that from~\cite[Lemma~3.5]{2017:Patel:poly_time_approx_partition_function} we can enumerate all polymers in $\mathcal C_\domainvar(e_0,k)$ in time $\mathcal O(k^5(\ex\Delta)^{2k})$. From Corollary~\ref{lem:connected_subgraph_bound_kappa} the number of weights we have to consider are at most $k(\kappa \ex\Delta)^k/2$. Given a polymer with $k$ edges we can compute its weight in polynomial time, hence we can compute all weights in time $\mathcal O(k^c(\kappa \ex\Delta)^k)$ for some constant $c$. Following the calculations in \cite[Lemma~16]{2019:Chen:fast_algorithms_low_temperature} we can show that the expected running time is constant.
\end{proof}

\begin{corollary}
	Let $G$ be a graph of maximum degree $\Delta$, let $\mathcal F\subseteq\mathcal F_0$ and let $\pi\colon V(G)\rightarrow \mathcal F$. Let $(\mathcal C_\domainvar(G),\Phi_\pi(\cdot,\mathbf{z}))$ be a polymer model satisfying the sampling condition. There is an $\varepsilon$-approximate sampling algorithm for $\mu_{\mathcal C_\domainvar(G)}$ with run-time in $\mathcal O(\Delta n\log(n/\varepsilon))$.
\end{corollary}
\begin{proof}
	The proof is identical to the proof of~\cite[Theorem~5]{2019:Chen:fast_algorithms_low_temperature}.
\end{proof}

In~\cite{2019:Chen:fast_algorithms_low_temperature} Chen et al. show how to use a sampler with simulated annealing in order to obtain a fast randomised algorithm (see~\cite[Section~3]{2019:Chen:fast_algorithms_low_temperature}). The runtime of this algorithm can be upper-bounded by the  $\mathcal O(n/\varepsilon^{2}\log(\Delta n/\varepsilon))$ calls to the randomised sampler. This approach can be used with the randomised sampler we describe above to get a randomised algorithm with runtime in $\mathcal O(\Delta n^2/\varepsilon^2\log^2(\Delta n/\varepsilon))$.

By randomised approximation algorithms we mean a \emph{fully polynomial-time randomised approximation scheme (FPRAS)}. An FPRAS is an algorithm that, given a problem and some fixed parameter $\varepsilon > 0$, produces with probability at least $3/4$ a multiplicative $\varepsilon$-approximation of its solution in time polynomial in in $n$ and $1/\varepsilon$.

{\renewcommand{\thetheorem}{\ref{thm:holant_poly_random}}
\begin{theorem}
 Let $G$ be a graph of maximum degree $\Delta$, $\mathcal F\subseteq\mathcal F_0$ and $\pi\colon V(G)\rightarrow \mathcal F$. For the Holant polynomial $Z_{\mathcal F}(G,\pi,\mathbf{z})$ there exists an
 $\varepsilon$-sampling algorithm from the distribution $\mu_G$ with run-time in $\mathcal O(\Delta n\log(n/\varepsilon))$ for $\mathbf{z}$ in
 \[\left\{(z_0,z_1,\dots,z_\domainvar)\in (\R_{\geq 0})^{\domainvar+1}\ \Big|\  z_0\neq0, \frac{z_i}{z_0}\leq ((\Delta\domainvar)^3\ex^5\maxValueOne^2)^{-1},\ 1\leq i\leq \domainvar\right\}\]
  with $\maxValueOne=\max\{1, r(\mathcal F)\}$. Furthermore, $Z_{\mathcal F}(G,\pi,\mathbf{z})$ admits an FPRAS with run-time $\mathcal O(\Delta n^2/\varepsilon^2\log^2(\Delta n/\varepsilon))$ for these values of $\mathbf{z}$.
\end{theorem}
}
\begin{proof}
Recall from the proof of Lemma~\ref{lem:Holant_polynomial_to_polymer_partition_function} that there is a bijection between the assignments $\sigma$ for the Holant polynomial and the families of polymers $\Gamma\in\mathcal I_\domainvar(G)$. Using this bijection, it suffices to sample from the equivalent distribution $\mu_{\mathcal C_\domainvar(G)}$. It remains to show that the mixing and the sampling conditions apply for our polymer model.

From the statement of the lemma, for any $\dompoly\in\mathcal C_\domainvar(G)$, we have 
\begin{align*}
\Phi_\pi(\dompoly,\mathbf{z})& = \left(\prod_{i=1}^\domainvar \left(\frac{z_i}{z_0}\right)^{|\adjkappa(\dompoly)|_i}\right)\prod_{v\in V(\gamma)} f_v(\adjkappa(\dompoly))\frac{1}{f_v(\mathbf{0})}\\
&\leq ((\Delta\domainvar)^3\ex^5\maxValueOne^2)^{-|E(\gamma)|}r(\mathcal F)^{|V(\gamma)|}.\\
\end{align*}

The graphs of our polymers are connected, hence for each $\dompoly\in\ \mathcal C_\domainvar(G)$, $|V(\gamma)|\leq |E(\gamma)|+1$. Together with the fact that $\maxValueOne=\max\{1, r(\mathcal F)\}$, this gives 
\begin{equation}\label{eq:sampling}
\Phi_\pi(\dompoly,\mathbf{z}) \leq ((\Delta\domainvar)^3\ex^5)^{-|E(\gamma)|}.
\end{equation}
Thus, this polymer system clearly satisfies the sampling condition (Definition~\ref{cond:sampling}).

For the mixing condition recall that for any polymer $\dompoly$, $|\mathcal C_\domainvar^{\dompoly}(i)|\leq|V(\gamma)|(\frac{(\Delta\domainvar\ex)^{i}}{2})$. Using the latter and~\eqref{eq:sampling}, we have
\begin{align*}
 \sum_{\dompoly'\not\sim\dompoly}|E(\dompoly')|\Phi_\pi(\dompoly',\mathbf{z})&\leq \sum_{i=1}^mi|V(\gamma)|(\frac{(\Delta\domainvar\ex)^{i}}{2})((\Delta\domainvar)^3\ex^5)^{-i}\\
 & \leq \sum_{i=1}^m\frac{|V(\gamma)|}{2}((\Delta\domainvar)^2\ex^3)^{-i}\\
 & <\frac{|V(\gamma)|}{2}.
\end{align*}
Since for any connected graph $\gamma$, $|V(\gamma)|\leq E(\gamma)+1$, the mixing condition (Definition~\ref{cond:mixing}) holds.
\end{proof}

Similar to the Holant polynomials, we can prove the following theorem for Holant problems.
{\renewcommand{\thetheorem}{\ref{thm:holant_random}}
\begin{theorem}
 Let $G$ be a graph of maximum degree $\Delta$, $\mathcal F\subseteq\mathcal F_0$ such that $r(\mathcal{F})\leq ((\Delta\domainvar)^{-\frac{3\Delta}{2}}\ex^{-\frac{5\Delta}{2}}) $ and $\pi\colon V(G)\rightarrow \mathcal F$. There exists an $\varepsilon$-sampling algorithm from the distribution $\mu_G$ for $Z_{\mathcal F}(G,\pi)$ with run-time in $\mathcal O(\Delta n\log(n/\varepsilon))$. Furthermore, $Z_{\mathcal F}(G,\pi)$ admits an FPRAS with run-time $\mathcal O(\Delta n^2/\varepsilon^2\log^2(\Delta n/\varepsilon))$.
\end{theorem}
}
\begin{proof}
 Following the proof of Theorem~\ref{thm:holant_poly_random}, we only have to show that the sampling and the mixing condition hold.
 
 From the statement of the lemma, for any $\dompoly\in\mathcal C_\domainvar(G)$, we have 
\[
\Phi_\pi(\dompoly,\mathbf{z}) = \prod_{v\in V(\gamma)} f_v(\adjkappa(\dompoly))\frac{1}{f_v(\mathbf{0})}
\leq r(\mathcal F)^{|V(\gamma)|}.
\]
From the conditions of the lemma, we have $r(\mathcal F)<1$. Recall that $G$ has maximum degree $\Delta$, thus for each $\dompoly\in$ $\gamma$, we have $|V(\gamma)|\geq (2|E(\gamma)|)/\Delta$ yielding
\[\Phi_\pi(\dompoly,\mathbf{z})\leq (r(\mathcal F))^{-\frac{2|E(\gamma)|}{\Delta}}\leq((\Delta\domainvar)^3\ex^5)^{-|E(\gamma)|},\]
which implies the sampling condition (Definition~\ref{cond:sampling}).

The proof for the mixing condition (Definition~\ref{cond:mixing}) is identical to the proof Theorem~\ref{thm:holant_poly_random}.
\end{proof}

%% file: outlook.tex

In this section we will show how the problem of counting weighted solutions to a system of sparse equations as studied by Barvinok and Regts~\cite{2019:Barvinok:integer_points} can be modelled as a polymer system and derive zero-free regions for this problem.
The parameter of this problem is a set of complex numbers $w_1,w_2,\dots,w_m\in\C$. We define the weight of a vector $\mathbf{x}\in\Z^m_{\geq0}$ with $\mathbf{x}=(x_1,x_2,\dots,x_m)$ as $w(\mathbf{x})=w_1^{x_1}w_2^{x_2}\dots w_m^{x_m}$. Given a finite set $X\subset\Z^m_{\geq0}$ we define its weight, 
\[
	w(X)=\sum_{\mathbf{x}\in X}w(\mathbf{x}).
\]

Given an $n\times m$ matrix $A=(a_{i,j})$ with $a_{i,j}\in\Z$ and positive integers $\kappa_1,\dots,\kappa_m\in\Z_{>0}$ the task is to approximate $w(X)$ with
\begin{equation}\label{eq:linear_subspace}
 X=\left\{\mathbf{x}\in\Z^m_{\geq 0}\mid Ax=\mathbf{0} \textrm{ and for all }j\in[m], \, x_j\leq \kappa_j\right\}.
\end{equation}

We view the above as a Holant problem where the input is a hypergraph. Here, the vectors correspond to edge assignments and the constraints given by each row of the matrix correspond to the vertex functions. We are going to model this problem as a polymer system. 

First, we explain the underlying hypergraph $H_A$ of the matrix~$A$. Given an $n\times m$ integer matrix $A=(a_{i,j})$, we define the hypergraph
$H_A$ with $n$ vertices and $m$ edges, where vertex $i$ is in edge $j$ if $a_{i,j}\neq 0$.
Hence computing $w(X)$ is a Holant problem on the hypergraph $H_A$, where edge $j$ can take an assignment in $[\kappa_j]$ and the signature of vertex $i$ outputs the value~1 if the constraint $\sum_{j=1}^m a_{i,j}x_j=0$ is satisfied. A vector $\mathbf{x}\in X$ corresponds to an assignment to the set of hyperedges of $H$. 

Given a vector $\mathbf{x}$, its support $\supp(\mathbf{x})=\{j\in[m]\mid x_j\neq0\}$ is the set of indices of its non-zero entries. Given a hypergraph $H$ and a hyperedge set $S\subseteq E(H)$, let $H[S]$ be the subhypergraph induced by the hyperedge set $S$. A hypergraph $H$ is connected if its underlying bipartite graph $G_H$ is connected, where $G_H=(V_1,V_2,E)$, with $V_1=V(H)$, $V_2=E(H)$ and $E=\{\{u,v\}\mid u\in V_1,v\in V_2\textrm{ and }v\in u \}$. 
We say that a vector $\mathbf{x}\in \Z^m_{\geq0}$ is connected if the subhypergraph $H[\supp(\mathbf{x})]$ induced by the hyperedge set in the support of $\mathbf{x}$ is connected.

The idea of our polymer system is to express the solutions in terms of their distance from $\mathbf{0}$, which is a solution to the system, and to represent this distance by a sum of connected vectors.
Hence, we define the set of polymers $\mathcal C(X)$ to contain the connected (with respect to $H_A$) vectors of $X\setminus\{\mathbf{0}\}$ and the weight of a polymer $\gamma\in\mathcal C(X)$ to be $\Phi(\gamma)=w(\gamma)$. The reason we exclude $\mathbf{0}$ from our polymer set is that this will be the ground state in our polymer representation. Two polymers $\gamma_1,\gamma_2\in\mathcal C(X)$ are defined to be incompatible, $\gamma_1\not\sim \gamma_2$, if their underlying hypergraphs share vertices, i.e.~$V(H_A[\supp(\gamma_1)])\cap V(H_A[\supp(\gamma_2)])\neq\emptyset$. The resulting families of pairwise compatible polymers denoted by $\mathcal I(X)$ contain then vectors with pairwise disjoint support. This is crucial to show that the sum over all vectors in such a family yields a vector in $X$.

\begin{lemma}\label{lem:points_polymers}
	For the polymer system $\mathcal C(X)$, we have
	\[Z(\mathcal C(X),\Phi(\cdot))=w(X).\]
\end{lemma}
\begin{proof}
	We first define the function $\sigma\colon\mathcal I(X)\rightarrow \Z^m_{\geq 0}$ with $\sigma(\emptyset)=\mathbf{0}$ and for each $\Gamma\in\mathcal I(X)$, $\sigma(\Gamma)=\sum_{\gamma\in\Gamma}\gamma$, where the summation of vectors is done element-wise. As we defined above, two compatible polymers $\gamma_1\sim\gamma_2$ cannot share vertices in the underlying hypergraph~$H_A$, which implies that $\supp(\gamma_1)\cap\supp(\gamma_2)=\emptyset$. Thus for a family of compatible polymers $\Gamma$ we have that $\sigma(\Gamma)\in X$ since there is no $j\in[m]$ with $\sigma(\Gamma)_j>\kappa_j$, and $A\sigma(\Gamma)=\mathbf{0}$ since for each polymer $\gamma\in\Gamma$ we have $A\gamma=\mathbf{0}$. We now show that $\sigma$ is bijective.

	For injectivity consider two families of compatible polymers $\Gamma_1\neq\Gamma_2$. Let $\gamma_1\in\Gamma_1\Delta\Gamma_2$ be a polymer with maximum support cardinality and assume without loss of generality that $\gamma_1\in\Gamma_1\setminus\Gamma_2$. First assume the existence of $j\in\supp(\gamma_1)$ that is not in the support of any polymer in $\Gamma_2$. The latter assumption implies $\sigma(\Gamma_1)\neq\sigma(\Gamma_2)$, since $(\sigma(\Gamma_1))_j\neq 0=(\sigma(\Gamma_2))_j$. 
	
	Now assume otherwise, i.e. $\supp(\gamma_1)\subseteq\bigcup_{\gamma\in\Gamma_2}\supp(\gamma)$. Since $\gamma_1$ is connected in $H_A$, the support of $\gamma_1$, which represents a set of connected hyperedges in $H_A$, cannot be partitioned into pairwise compatible sets. Hence, $\Gamma_2$ can contain at most one polymer $\gamma_2$ with $\supp(\gamma_1)\subseteq\supp(\gamma_2)$. Due to the maximality of the support cardinality in the choice of $\gamma_1$ we have $\supp(\gamma_1)=\supp(\gamma_2)$. Since $\gamma_1\notin\Gamma_2$, this means that $\gamma_1\neq\gamma_2$ which means that there is at least one index~$j$ in which $\gamma_1$ and $\gamma_2$ differ. By the definition of incompatibility, no other polymer in $\Gamma_1$ or $\Gamma_2$ can affect this index~$j$ and it follows that $(\sigma(\Gamma_1))_j=(\sigma(\{\gamma_1\}))_j\neq(\sigma(\{\gamma_2\}))_j=(\sigma(\Gamma_2))_j$. This implies $\sigma(\Gamma_1)\neq\sigma(\Gamma_2)$ and concludes the proof that $\sigma$ is injective.
	
	For surjectivity we will argue that, given a vector $x\in X$, there is always a family of polymers $\Gamma\in \mathcal I(X)$ with $\sigma(\Gamma)=x$. Assume that the hypergraph $H_A[\supp(x)]$ induced by the support of $x$ has $k$ connected components induced by the sets of edges $C_1\dots C_k$ for some $k\geq0$. We deduce that the vectors $\gamma_i$ with $(\gamma_i)_j=x_j$ for all $j\in C_i$ and $(\gamma_i)_j=0$ otherwise are in $\mathcal C(X)$ since by definition their support induces connected subhypergraphs of $H_A$. Hence, for the family $\Gamma=\{\gamma_1,\dots\gamma_k\}$ we have $\sigma(\Gamma)=x$.
	
	Finally, observe that for any family $\Gamma\in\mathcal I(X)$ we have $\prod_{\gamma\in\Gamma}w(\gamma)=w(\sigma(\Gamma))$. The latter yields
	\begin{align*}
		Z(\mathcal C(X),\Phi(\cdot))& = \sum_{\Gamma\in\mathcal I(X)}\prod_{\gamma\in\Gamma}w(\gamma)\\
		&=\sum_{\Gamma\in\mathcal I(X)}w(\sigma(\Gamma))\\
		&=\sum_{x\in X}w(x) = w(X),
	\end{align*}
		where the last equality comes from the bijectivity of $\sigma$.
	\end{proof}

By this polymer representation we can now apply Theorem~\ref{thm:cluster_expansion_polymer_models} to prove absolute convergence of $\log w(X)$ which yields the following.
\begin{lemma}\label{lem:points}
	Let $A$ be an $n\times m$ integer matrix, let $\kappa_1,\dots,\kappa_m\in\Z_{>0}$ and $X$ as in \eqref{eq:linear_subspace}. Assume the number of non-zero entries in every row of $A$ does not exceed $r$ for some $r\geq2$ and the number of non-zero entries in every column of $A$ does not exceed $c$ for some $c\geq 1$. Let $w=\max\{|w_1|, \dots, |w_m|\}$ and $\kappa=\max\{\kappa_1, \dots, \kappa_m\}$. If $w\leq ({(r\ex+1)c\kappa\ex^{1/2}})^{-1}$, then $\log w(X)$ converges absolutely.
\end{lemma}
\begin{proof}
Consider the polymer representation of $w(X)$ from Lemma~\ref{lem:points_polymers}. Based on the definition of $w$ and $w\leq 1$ we deduce that for every $\gamma\in\mathcal C(X)$ it holds $w_j^{x_j}\leq w$ for all $j\in\supp(\gamma)$. Thus, $\Phi(\gamma)=w(\gamma)\leq w^{|E(\gamma)|}$.
To apply Theorem~\ref{thm:cluster_expansion_polymer_models} consider an arbitrary fixed polymer $\gamma\in \mathcal C(X)$ and denote by $\mathcal C^\gamma(i) = |\{\gamma'\in \mathcal C (X)\mid \gamma'\nsim\gamma, |E(\gamma)|=i \}|$ the number of polymers of size $i$ incompatible with $\gamma$. As in our other bounds, we choose $a(\gamma') = \alpha|E(\gamma')|$ where $\alpha\in\mathbb R_{>0}$ will be chosen later in order to optimise the bound.
\begin{align*}
	\sum_{\gamma'\nsim\gamma} |\Phi(\gamma')| \ex^{a(\gamma')} &\leq \sum_{\gamma'\nsim\gamma} w^{|E(\gamma')|} \ex^{a(\gamma')} \\
	& \leq \sum_{i=1}^{m} \mathcal C^\gamma(i) w^i \ex^{ \alpha i}.
\end{align*}
	
Observe that in our hypergraph representation $r$ denotes the maximum degree of $H_A$ and $c$ denotes the maximum hyperedge size of $H_A$. Further observe that a hypergraph $H=(V,E)$ is also a hypergraph $H'=(E,V)$ with the same underlying bipartite graph representation, i.e., $G_H=G_{H'}$. Now we can use~\cite[Corollary~3.8]{Liu_2018} to deduce that there are at most $\frac{(erc)^{i-1}}{2}\kappa^i$ many edge coloured connected hyperedge-induced subhypergraphs of $H_A$ with $i$ hyperedges that contain a fixed hyperedge. The latter corollary applies as coloured hypergraphs fall into the definition of \emph{insects}~\cite[Definition~3.8]{Liu_2018}.
This bound implies $\mathcal C^\gamma(i)\leq \frac{(\ex rc)^{i-1}}{2}\kappa^i |V(\gamma)|$, which gives
\begin{align*}
	\sum_{\gamma'\nsim\gamma} |\Phi(\gamma')| \ex^{a(\gamma')} &\leq \sum_{i=1}^{m}\frac{(\ex rc)^{i-1}}{2}\kappa^i |V(\gamma)|w^i \ex^{ \alpha i}\\
	 &\leq \frac {|V(\gamma)|}{2rc\ex}\sum_{i=1}^{m}(\ex^{\alpha+1}rc\kappa w)^i .
\end{align*}

Applying Theorem~\ref{thm:cluster_expansion_polymer_models} it remains to identify the values of $w$ for which the following holds,
\[
	\frac {|V(\gamma)|}{2rc\ex}\sum_{i=1}^{m}(\ex^{\alpha+1}rc\kappa w)^i\leq \alpha|E(\gamma)|.
\]

Since $m$ can be arbitrarily large, we can use the geometric series formula for the sum on the left-hand side of the above inequality, hence the above equation is implied by
\begin{align*}
&&\sum_{i=0}^{\infty}(\ex^{\alpha+1}rc\kappa w)^i -1&\leq 2rc\ex\alpha\frac{|E(\gamma)|}{|V(\gamma)|} \\
\Leftarrow\quad && \frac{1}{1-\ex^{\alpha+1}rc\kappa w}-1&\leq 2rc\ex\alpha\frac{|E(\gamma)|}{|V(\gamma)|}.
\end{align*}

We know that in $H$, each hyperedge contains at most $c$ vertices, therefore for any polymer~$\gamma$ we have $|V(\gamma)|\leq c|E(\gamma)|$. Using this the above inequality is implied by
\begin{align*}
  && \frac{1}{1-\ex^{\alpha+1}rc\kappa w}-1 &\leq 2r\ex\alpha\\
  \Leftrightarrow\quad\quad && \ex^{\alpha+1}rc \kappa w &\leq\frac{2r\ex\alpha}{2r\ex\alpha+1} \\
  \Leftrightarrow\quad\quad && w &\leq \frac{2\alpha}{(2r\ex\alpha+1)c\kappa \ex^\alpha}.
\end{align*}
Through derivative analysis, we find this is maximal for $2\ex r\alpha^2+\alpha-1=0$, which resolves to $\alpha = \frac{\sqrt{8\ex r+1}-1}{4\ex r}$.
This brings the final bound to
\[
	w\leq \frac{\sqrt{8\ex r+1}-1}{(\sqrt{8\ex r+1}+1)rc\kappa\ex^{(\sqrt{8\ex r+1}-1)/(4\ex r)+1}}.
\]

Choosing $\alpha=\frac{1}{2}$ for simplicity instead yields the bound given in the statement
\[
	w\leq \frac{1}{(r\ex+1)c\kappa\ex^{1/2}}.
\]
\end{proof}

When $X$ is restricted to vectors in $\{0,1\}^m$ Barvinok and Regts~\cite[Theorem 1.1]{2019:Barvinok:integer_points} using a different technique were able to obtain the improved bound of $w\leq\frac{0.46}{r\sqrt{c}}$. Furthermore given an integer $\kappa>1$ for $X\subseteq (\mathbb Z/\kappa\mathbb Z)^m$, \cite[Theorem 1.2]{2019:Barvinok:integer_points} yields the bound $w\leq\frac{0.46}{(\kappa-1)r\sqrt{c}}$.
Our bounds for $w(X)$ yield smaller regions than the theorems of Barvinok and Regts but are applicable to a broader class of sets $X$. 
As a possible application of their theorem, they give a polynomial expressing perfect matchings in terms of their distance from a given perfect matching. As we show in the next section, we can abuse the type of constraints perfect matchings impose and obtain a better bound by using our polymer approach.

\subsection{Improved zero-free regions for a perfect matching polynomial}
\label{sec:hpm}

Given a $k$-hypergraph $H=(V,E)$ and a perfect matching $\groundmatching\subseteq E$ one can define a partition function for perfect matchings by expressing the distance to $\groundmatching$ in the following way
\[
	Z_\mathrm{pm}(H,\groundmatching,z)=\sum_{\sigma\in\{0,1\}^E}\prod_{v\in V(G)} f(\sigma|_{\adjedges(v)})z^{|\sigma\difference \sigma_\groundmatching|}.
\]
Here, $\sigma_\groundmatching\in\{0,1\}^E$ denotes the signature corresponding to $\groundmatching$, $\sigma\difference\bar\sigma=\{e\in E\mid \sigma(e)\neq \bar\sigma(e)\}$ for any two assignments $\sigma, \bar \sigma\in\{0,1\}^E$ and $f$ is the ``exactly-one''-function. Note that finding a perfect matching on a $k$-hypergraph with $k\geq3$ is $\np$-complete (see e.g.~\cite{1999:Ausiello:Complexity_and_approximation}), therefore we assume that $\groundmatching$ is given. Barvinok and Regts~\cite{2019:Barvinok:integer_points} prove that for $|z|\leq \frac{0.46}{\Delta\sqrt{k}}$ this polynomial has no roots. By using our technique we can improve this bound to $|z|\leq((\Delta-1+k)e)^{-1}$.

Let $H=(V,E)$ be a $k$-hypergraph of max degree $\Delta$. Let $G=(V_L, V_R,E)$ be the corresponding bipartite graph, where $V_L=V(H)$, $V_R=E(H)$ and $E(G)=\{(u,e)\mid u\in V(H),e\in E(H) \textrm{ and } u\in e \}$. Let $\Delta_L=\Delta$ be the maximum degree of the vertices in $V_L$ and $\Delta_R=k$ be the maximum degree of the vertices in $V_R$.
Any perfect matching $M'$ in $H$ can be seen as an assignment $\sigma_{M'}:V_R\rightarrow \{0,1\}$, with $\sigma_{M'}(v)=1$ if and only if $v\in M'$. Let $\groundmatching$ be the ``ground'' perfect matching for $H$. Hence $\groundmatching\subseteq V_R$ in the bipartite graph $G$. 

Given a perfect matching $M'$, let $E_{M'}=\{v\in V_R\mid \sigma_{\groundmatching}(v)\neq \sigma_{M'}(v)\}$ and consider the subgraph~$G[S_{M'}]$ induced by the set $S_{M'}=E_{M'}\cup\bigcup_{v\in E_{M'}}\Gamma_G(v)$. 
We can rewrite $Z_\mathrm{pm}$ in terms of the bipartite graph $G$ as

\[
	Z_\mathrm{pm}(G,\groundmatching,z)=\sum_{M'\in\mathcal{PM}(H)} z^{|V_R(G[S_{M'}])|},
\]
where $\mathcal{PM}(H)$ denotes the set of perfect matchings of $H$. Observe that for each $M'\in\mathcal{PM}(H)$, $G[S_{M'}]$ has the property that each vertex $v\in V(G[S_{M'}])\cap V_L$ has exactly one neighbour in $M$ and one neighbour in $ V_R\setminus M$.

We are now ready to define a polymer model for $Z_\mathrm{pm}(G,z)$. The set of polymers $\mathcal C(G)$ contains the connected subgraphs of $G$ such that every vertex in $V_L$ has exactly one neighbour in $M$ and exactly one neighbour in $V_R\setminus M$ and that for every $v\in V_R$, $\Gamma_G(v)\subseteq V(\gamma)$. We say that two polymers $\gamma\not\sim\gamma'$ are incompatible if $V(\gamma)\cap V(\gamma')\neq\emptyset$. Given a polymer $\gamma$ and $z\in\mathbb C$, we define $\Phi_\mathrm{pm}(\gamma,z)=z^{|V_R(\gamma)|}$. 

\begin{lemma}
 Let $H=(V,E)$ be a finite hypergraph and let $G$ be the corresponding bipartite graph. $Z_\mathrm{pm}(H,M,z)=Z(\mathcal C(G),\Phi_\mathrm{pm}(\cdot,z))$ for all $z\in \mathbb C$.
\end{lemma}
\begin{proof}
Since $Z_\mathrm{pm}(H,M,z)=Z_\mathrm{pm}(G,M,z)$, it suffices to show that $Z_\mathrm{pm}(G,M,z)=Z(\mathcal C(G),\Phi_\mathrm{pm}(\cdot,z))$. We now give a bijection between the perfect matchings $M'$ and the families of polymers $\Gamma$, such that $z^{|V_R(G[S_{M'}])|}=\prod_{\gamma\in\Gamma}\Phi_\mathrm{pm}(\gamma,z)$. 

The ground matching $M$ corresponds to the empty set. Given a perfect matching $M'$ observe that the connected components of $G[S_{M'}]$ are in fact polymers in $\mathcal C(G)$, say $\gamma_1,\gamma_2,\dots,\gamma_j$, where $j\geq0$. Since $\gamma_1,\gamma_2,\dots,\gamma_j$ are the connected components of $G[S_{M'}]$, they are pairwise compatible. From the definition of $\Phi_\mathrm{pm}(\gamma,z)$, it follows that $\prod_{i=1}^j\Phi_\mathrm{pm}(\gamma_i,z)=\prod_{i=1}z^{|V_R(\gamma_i)|} =z^{|V_R(G[S_{M'}])|}$. 
\end{proof}

In order to apply Theorem~\ref{thm:cluster_expansion_polymer_models} we will use the following lemma.

\begin{lemma}\label{lem:bipartite_graph_enumeration}
 Let $G=(V_L, V_R,E)$ be a bipartite graph, where $\Delta_L$ is the maximum degree of the vertices in $V_L$, such that the vertices in $V_R$ have degree at least 2.  Let $M\subseteq V_R$, such that every vertex $v\in V_L$ has exactly one neighbour in $M$. Given a vertex $u\in V_L $, the number of connected induced subgraphs~$G'$ of $G$ such that the following hold:
 \begin{enumerate}
  \item $u\in V(G')$;
  \item $|V(G')\cap V_R|=k$;
  \item for each $v\in V_R$, $\Gamma_G(v)\subseteq V(G')$; and
  \item each $v\in (V(G')\cap V_L)$ has exactly one neighbour in $V(G')\cap M$ and exactly one neighbour in $V(G')\cap (V_R\setminus M)$
 \end{enumerate}
is at most $(\Delta_L-1)^{k-1}$.
\end{lemma}
\begin{proof}
Let $V_L'=V(G')\cap V_L$ and $V_R'=V(G')\cap V_R$. Since $u\in V_L$, we have $u\in V_L'$. Condition 4 implies that in $G'$, $u$ has exactly two neighbours $v_M\in M$ and $v_{\bar{M}}\notin M$. By the definition of $M$, the choice of $v_M$ is fixed and there are at most $\Delta_L-1$ choices for $v_{\bar{M}}\in V_R\setminus M$. Since $G'$ is connected, either $k=2$ and the lemma is proved, or there must exist at least one other vertex $u_1\in\Gamma_{G'}(v_M)\cup\Gamma_{G'}(v_{\bar{M}})$. As in the case of $u$, Condition 4 implies the existence of a fixed neighbour of $u_1$ in $M$ and at most $\Delta_L-1$ choices for a neighbour of $u_1$ in $V_R\setminus M$. Proceeding with this argumentation for each vertex in $V_L'$ we can observe that each vertex in $V_R'\setminus M$ is a result of one of the up to $\Delta_L-1$ choices for each vertex in $V_L'$. Since there must be at least one vertex $v\in M$ whose choice is fixed, the number of the induced subgraphs $G'$ is at most $(\Delta_L-1)^{k-1}$.
\end{proof}


\begin{lemma}
 For all $k$-hypergraphs~$H$ with maximum degree $\Delta$ and given a ground perfect matching $\groundmatching$, the polynomial $Z_\mathrm{pm}(H,M,z)$ has no roots in $|z|\leq((\Delta-1+k)e)^{-1}$.
\end{lemma}
\begin{proof}
 Let $H$ be a $k$-hypergraph with $G=(V_L,V_R,E)$ its corresponding bipartite graph, where $\Delta_L=\Delta$ and $\Delta_R=k$. Consider the polymer model with $Z(\mathcal C(G),\Phi_\mathrm{pm}(\cdot,z))=Z_\mathrm{pm}(H,M,z)$. As in all our absolute convergence proofs we apply Theorem~\ref{thm:cluster_expansion_polymer_models} for $Z(\mathcal C(G),\Phi_\mathrm{pm}(\cdot,z))$ by choosing $a(\gamma)=\alpha|V_R(\gamma)|$, where $\alpha\in\mathbb{R}_{\geq0}$ to be chosen later. Let $\mathcal C_\mathrm{pm}^\gamma(i)$ be the number of polymers $\gamma'\not\sim\gamma$ with $|V_R(\gamma')|=i$. We have
 
\begin{align*}
	\sum_{\gamma'\nsim\gamma} |\Phi(\gamma')| \ex^{a(\gamma')} &\leq \sum_{\gamma'\nsim\gamma}{|z|}^{|V_R(\gamma')|} \ex^{\alpha|V_R(\gamma')|} \\
	& \leq \sum_{i=1}^{|V_R|} \mathcal C_\mathrm{pm}^\gamma(i) {|z|}^i \ex^{ \alpha i}.
\end{align*}
	
We observe that each polymer $\gamma$ is a connected induced subgraph of $G$ that fulfills Conditions~3 and~4 of the statement of Lemma~\ref{lem:bipartite_graph_enumeration}, with $M=M_0$. Therefore, by Lemma~\ref{lem:bipartite_graph_enumeration} we have $\mathcal C_\mathrm{pm}^\gamma(i)\leq |V_L(\gamma)|(\Delta_L-1)^{i-1}$. Therefore, to prove the lemma it suffices to show that
 \begin{align*}
 & \sum_{i=1}^{|V_R|}|V_L(\gamma)|(\Delta_L-1)^{i-1}{|z|}^i \ex^{ \alpha i} \leq\alpha|V_R(\gamma)|\\
 \Leftrightarrow\quad\quad & \frac{|V_L(\gamma)|}{\Delta_L-1}\sum_{i=1}^{|V_R|}(\Delta_L-1)^{i}{|z|}^i\ex^{ \alpha i}\leq\alpha|V_R(\gamma)|.
 \end{align*}
 Since $\Delta_R$ is the maximum degree of each vertex in $V_R$, for each polymer $\gamma$, $\frac{1}{\Delta_R}\leq \frac{|V_R(\gamma)|}{|V_L(\gamma)|}$. It remains to show
 \[ \sum_{i=0}^{|V_R|}(\Delta_L-1)^{i}{|z|}^i\ex^{ \alpha i}\leq\frac{\alpha(\Delta_L-1)+\Delta_R}{\Delta_R}.\]
 
Since $|V_R|$ might be arbitrarily large, using the formula for geometric series, the above inequality is implied by
 \begin{align*}
  & \frac{1}{1-(\Delta_L-1)e^\alpha {|z|}}\leq \frac{\alpha(\Delta_L-1)+\Delta_R}{\Delta_R}\\
  \Leftrightarrow\quad\quad & (\Delta_L-1)e^\alpha {|z|}\leq 1-\frac{\Delta_R}{\alpha(\Delta_L-1)+\Delta_R}\\
  \Leftrightarrow\quad\quad & (\Delta_L-1)e^\alpha {|z|}\leq \frac{\alpha(\Delta_L-1)}{\alpha(\Delta_L-1)+\Delta_R}\\
  \Leftrightarrow\quad\quad & {|z|}\leq \frac{\alpha}{(\alpha(\Delta_L-1)+\Delta_R)\ex^\alpha}.
 \end{align*}
Substituting $\Delta=\Delta_L$ and $k=\Delta_R$, we obtain
\[{|z|}\leq \frac{\alpha}{(\alpha(\Delta-1)+k)\ex^\alpha}.\]

The optimal bound can be found by maximizing the function $f(\alpha)=\frac{\alpha}{(\alpha(\Delta-1)+k)\ex^\alpha}$. We observe that
\[f'(\alpha)=\frac{(\alpha(\Delta-1)+k)\ex^\alpha-\alpha((\Delta-1)\ex^\alpha+(\alpha(\Delta-1)+k)\ex^\alpha)}{((\alpha(\Delta-1)+k)\ex^\alpha)^2}.\]

Solving $f'(\alpha)=0$ for $\alpha>0$ yields $\alpha=\frac{\sqrt{k^2+4(\Delta-1)k}-k}{2(\Delta-1)}$ and, furthermore, $f''(\alpha)<0$ at this value.

Substituting we obtain the bound
\[{|z|}\leq\frac{\sqrt{k^2+4(\Delta-1)k}-k}{(\Delta-1)(\sqrt{k^2+4(\Delta-1)k}+k)\ex^\frac{\sqrt{k^2+4(\Delta-1)k}-k}{2(\Delta-1)}}.\]

Setting $\alpha=1$ for simplicity yields the bound $|z|\leq((\Delta-1+k)e)^{-1}$ given in the statement of the lemma.
\end{proof}

\subsubsection{Perfect matchings on graphs}

\label{sec:pm}

%
When $H=G$ is a graph of maximum degree $\Delta$, instead of a hypergraph, we obtain even better bounds for $Z_\mathrm{pm}(G,M,z)$. Since perfect matchings on graphs $G$ are Holants, recall the definition of our polymer system we employed for the study of the Holant framework. That is, $\mathcal C(G)$ is the set of connected subgraphs containing at least one edge and for two polymers $\gamma, \gamma' \in \mathcal C(G)$ we defined $\gamma \not\sim \gamma'$ if and only if $V(\gamma) \cap V(\gamma') \neq \emptyset$. Furthermore, we identified an assignment $\sigma \in \{0,1\}^E$ of edges with a vector in $\{0,1\}^{|E|}$ by assuming an inherent enumeration of the edges. In fact, we can use this model to consider the distance to a ground state different from the empty set. To this end, for $v \in V$ we define the vertex constraints $f_v$ to pick only assignments corresponding to cycles alternating between the perfect matching $M$ and the assigned edgeset, where we also allow an empty cycle. Formally, $f_v$ assigns~$1$ to $\mathbf{0}$ as well as to vectors $\mathds{1}_{e_M}+ \mathds{1}_e$ for $e_M\in M$ and $e\in E(G)\setminus M$ with $v\in e_M\cap e$, and~$0$ otherwise.
We define the weight function for a polymer $\gamma\in \mathcal C(G)$ and $z\in \mathbb C$ by
\[
	\Phi_\mathrm{pm}(\gamma, z)=\prod_{v\in V(\gamma)} f_v(\mathds{1}_{E(\gamma)})z^{|E(\gamma)|}.
\]

\begin{lemma}
	$Z_\mathrm{pm}(G,z)=Z(\mathcal C(G),\Phi_\mathrm{pm}(\cdot,z))$ for all finite graphs~$G=(V,E)$  and all $z\in \mathbb C$.
\end{lemma}
\begin{proof}
	Every perfect matching corresponds to an assignment $\sigma\in\{0,1\}^E$ with $f(\sigma)=1$ that contributes $z^{\sigma\difference \sigma_\groundmatching}$ to the partition function $Z_\mathrm{pm}(G,z)$. The difference $\sigma\difference \sigma_\groundmatching$ corresponds to a set of disjoint cycles, alternating with respect to $\groundmatching$. In this way, any perfect matching $\sigma$ translates to a set of polymers, each  of non-zero weight, which together yield the value $z^{\sigma\difference \sigma_\groundmatching}$. 
	
	Conversely, an incompatible set of polymers $\Gamma\in \mathcal I(G)$ of non-zero weight assigned by $\Phi_\mathrm{pm}$ translates to a perfect matching by editing the ground matching $\groundmatching$ according to the edges altered by the polymer as follows. The set $\Gamma$ corresponds to the perfect matching $\groundmatching_\Gamma=(\groundmatching\setminus(\groundmatching\cap E(\Gamma)))\cup (E(\Gamma)\setminus \groundmatching)$ which differs from $\groundmatching$ by $ |\groundmatching\cap E(\Gamma))|+|E(\Gamma)\setminus \groundmatching|=|E(\Gamma)|$ edges, so $\prod_{\gamma \in \Gamma} \Phi_\mathrm{pm}(\gamma,z) =z^{|E(\Gamma)|}=z^{\sigma_\groundmatching\difference\sigma_{\groundmatching_\Gamma}}$.
	
	In both cases, the empty family $\{\emptyset\}\in\mathcal I (G)$ corresponds to the assignment $\sigma=\mathbf{0}$ which counts 1 to represent the contribution of the ground polymer.
\end{proof}

\begin{theorem}\label{thm:perfect_matchings}
	For all graphs $G$ of maximum degree $\Delta$ the polynomial $Z_\mathrm{pm}(G,z)$ has no roots in $\{z\in \mathbb C\mid |z|\leq \left(\sqrt{4.85718 \, (\Delta -1)}\right)^{-1}\}$ and the cluster expansion of $\log Z_\mathrm{pm}(G,z)$ is converging absolutely.
\end{theorem}
\begin{proof}
	We will again apply Theorem~\ref{thm:cluster_expansion_polymer_models} to the polymer representation $Z(\mathcal C(G),\Phi_\mathrm{pm}(\cdot,z))$. For the sake of simplicity, we now restrict this polymer definition to the set of polymers with non-zero contribution to the partition function. Hence, let 
	\[
	\mathcal C_\mathrm{pm}(G)=\{\gamma\in \mathcal C(G)\mid\prod_{v\in V(\gamma)}f_v(\mathds{1}_{E(\gamma)})=1\}
	\] 
	and obviously $Z(\mathcal C(G),\Phi_\mathrm{pm}(\cdot,z))=Z(\mathcal C_\mathrm{pm}(G),\Phi_\mathrm{pm}(\cdot,z))$. 
	
	We fix a polymer $\gamma\in\mathcal C_\mathrm{pm}(G)$ and in order to bound the number of polymers incompatible with $\gamma$ we denote by $\mathcal C_\mathrm{pm}^\gamma(i)$ the number of polymers $\gamma'\not\sim\gamma$ with $|E(\gamma')|=i$. First, observe that each polymer in $\mathcal C_\mathrm{pm}(G)$ has an even number of edges, and more precisely, $|E(\gamma')\cap\groundmatching|=\frac12|E(\gamma')|$. Further, by the special structure of our polymers in $\mathcal C_\mathrm{pm}(G)$ we can estimate $|\mathcal C_\mathrm{pm}^{\gamma}(i)|$ better than with the general Lemma~\ref{lem:connected_subgraph_bound}.
	
	Any $\gamma'\in\mathcal C_\mathrm{pm}(G)$ with $\gamma'\not\sim\gamma$ has to intersect with $\gamma$ on at least one of the edges from $\groundmatching$. There are exactly $\frac 12 |E(\gamma)|$ choices for such a common edge in $\groundmatching\cap E(\gamma)\cap E(\gamma')$. Starting with such a fixed edge from $\groundmatching$, $\gamma'$ builds a cycle alternating between edges from $E(G)\setminus \groundmatching$ and $\groundmatching$. As $\groundmatching$ is a perfect matching the choice of edges from $E(G)\setminus \groundmatching$ completely determines the edges from $\groundmatching\cap E(\gamma')$. Moreover, in a graph of maximum degree $\Delta$ there are only $\Delta-1$ choices to continue from a fixed vertex with edges from $E(G)\setminus \groundmatching$. This yields a total of at most $\frac 12 |E(\gamma)|(\Delta-1)^{i/2}$ polymers in $\mathcal C_\mathrm{pm}^\gamma(i)$.
	
	We now set $a(\gamma)=\alpha |E(\gamma)|$, where we will choose $\alpha > 0$ later in order to improve the bound on the region. By this choice of $a(\gamma)$ and with $j=\frac{i}{2}$ we obtain
	\begin{align*}
		\sum_{\gamma' \not\sim \gamma} |\Phi_\mathrm{pm}(\gamma',{z})| \ex^{a(\gamma')}
		\leq & \sum_{j=2}^{\frac12 |V(G)|} |\mathcal C_\mathrm{pm}^{\gamma}(2j)|\,|z|^{2j}\,\ex^{2\alpha j}
		\\
		\leq &  \sum_{j=2}^{\frac12 |V(G)|} \frac 12 |E(\gamma)|(\Delta-1)^{j}\,|z|^{2j}\,\ex^{2\alpha j } .
	\end{align*}
	In order to apply Theorem~\ref{thm:cluster_expansion_polymer_models}, we are going to choose a bound on $|z|$ such that
	\[
		\sum_{j=2}^{\frac12 |V(G)|} \frac 12 |E(\gamma)|(\Delta-1)^{j}\,|z|^{2j}\,\ex^{2 \alpha j} \leq a (\gamma) =  \alpha |E(\gamma)| ,
	\]
	which is equivalent to
	\[
		\sum_{j=2}^{\frac12 |V(G)|} (\Delta-1)^{j}\,|z|^{2j}\,\ex^{2 \alpha j} \leq 2 \alpha .
	\]	
	We now set
	$
		q := (\Delta-1) \,  |z|^{2} \, \ex^{2 \alpha},
	$
	which is equivalent to 
	$
		|z| = \sqrt{\frac{q}{(\Delta-1)\ex^{2 \alpha} }},
	$
	and by this choice and the geometric series we obtain
	\[
		\sum_{j=2}^{\frac12 |V(G)|} (\Delta-1)^{j}\,|z|^{2j}\,\ex^{2 \alpha j} = \sum_{j=2}^{\frac12 |V(G)|} q ^j < \frac{1}{1-q} - 1
	\]	
	assuming that $q\in (0, 1)$.
	Hence, we have to choose $\alpha$ such that
	$
		\frac{1}{1-q} - 1 \leq 2 \alpha ,
	$
	which holds exactly if $q\leq \frac{2 \alpha}{1 + 2 \alpha}$. We insert this choice of $q$ into the chosen bound on $|z|$ and obtain
	\begin{equation}\label{eq:perfect_matching_optimal_alpha}
		|z| \leq \sqrt{\frac{2\alpha}{1+2\alpha} \, \frac{1}{(\Delta-1)\ex^{2 \alpha} }}.
	\end{equation}
	Define $g(\alpha)=\sqrt{\frac{2\alpha}{1+2\alpha} \, \frac{1}{(\Delta-1)\ex^{2 \alpha} }}$. The derivative of $g$
	has a positive root at $\alpha = \frac{-1 + \sqrt{5}}{4}$ and, furthermore, $f''(\alpha) <0$ at this point.
	
	Therefore, choosing $\alpha = \frac{-1 + \sqrt{5}}{4}$ maximises the bound of \eqref{eq:perfect_matching_optimal_alpha} to
	
	\[
		|z| \leq \sqrt{\frac{3- \sqrt{5}}{2(\Delta-1)\left(\ex^{\frac{\sqrt5 -1}{2}}\right)}} \approx \left(\sqrt{4.85718 \, (\Delta -1)}\right)^{-1}.
	\]
\end{proof}